\documentclass[preprint,12pt,authoryear]{elsarticle}

\usepackage{amssymb}
\usepackage{amsmath}
\usepackage{subcaption}

\allowdisplaybreaks

\usepackage{amsthm} % Loads theorem environments

\newtheorem{theorem}{Theorem}[section]
\newtheorem{definition}[theorem]{Definition}

\newtheorem{corollary}[theorem]{Corollary}

\usepackage[ruled,vlined]{algorithm2e}

\usepackage[colorlinks=true]{hyperref}
\usepackage[margin=1in]{geometry}
\makeatletter
\def\ps@pprintTitle{%
 \let\@oddfoot\@empty
 \def\@oddfoot{\hfil \thepage\hfil}  % Centered page number in footer
}
\makeatother

\begin{document}

\begin{frontmatter}

\title{\textbf{Explainable Outlier Detection for Interval-valued Data}}

\author[inst1]{Catarina P. Loureiro\corref{cor1}}
\ead{catarinapadrela@tecnico.ulisboa.pt}
\author[inst1]{M. Ros\'ario Oliveira}
\ead{rosario.oliveira@tecnico.ulisboa.pt}
\author[inst2]{Paula Brito}
\ead{mpbrito@fep.up.pt}
\author[inst3]{Lina Oliveira}
\ead{lina.oliveira@tecnico.ulisboa.pt}

\cortext[cor1]{Corresponding author.}

\affiliation[inst1]{organization={CEMAT and Department of Mathematics, Instituto Superior Técnico, Universidade de Lisboa},
            country={Portugal}}

\affiliation[inst2]{organization={Faculdade de Economia, Universidade do Porto, and LIAAD-INESC TEC},
            country={Portugal}}

\affiliation[inst3]{organization={CAMGSD and Department of Mathematics, Instituto Superior Técnico, Universidade de Lisboa},
            country={Portugal}}

\begin{abstract}
Explainability is increasingly recognized as a key aspect of outlier detection. However, for complex data structures such as interval-valued data, it remains largely unexplored. Building on an outlier detection framework based on the Interval Minimum Covariance Determinant estimator, we propose a novel approach to explain the outlyingness of interval-valued observations using the concept of the Shapley value. We derive a closed-form expression for the Shapley value of the squared robust Interval-Mahalanobis distance, enabling efficient computation of variable contributions. This formulation allows for a fine-grained interpretation of outliers, providing a detailed decomposition into contributions from centers, ranges, and cross-terms of the interval-valued observations. Moreover, the Shapley value is closely connected to the concept of cellwise outliers, as it can help identify variable-specific outliers that may not be evident at multivariate level. We further extend the framework through the Shapley interaction index to capture pairwise variable interactions driving atypical behavior. The practical utility of the proposed approach is illustrated through two real-world datasets.
\end{abstract}

\begin{keyword}
Symbolic Data Analysis \sep Shapley Value \sep Robust Interval-Mahalanobis Distance \sep Cellwise Outliers \sep Interval Minimum Covariance Determinant Estimator
\end{keyword}

\end{frontmatter}

\section{Introduction}
\label{sec:intro}
The importance of explainability in Artificial Intelligence (AI) and data analysis has become increasingly recognized in recent years. As Machine Learning (ML) and statistical methods become more complex and are deployed in high-stakes applications, the demand for interpretable and trustworthy models has grown substantially, resulting in the rise of eXplainable AI (XAI) (see, e.g, \cite{XAI2020,XAI2024}). This demand extends naturally to outlier detection, where not only identifying atypical observations but understanding why they are considered atypical is crucial for decision-making, scientific insight, or stakeholder confidence \citep{li2023survey}. Outlier detection plays a vital role in data quality assessment, fraud detection, medical diagnosis, and numerous other domains. While the detection of outliers is important, equally important is the ability to explain which features drive the atypical behavior of an observation. This explainability becomes particularly challenging when dealing with non-standard data structures, such as interval-valued data, which are increasingly common in modern applications.

Symbolic Data Analysis (SDA) \citep{Diday} has become a powerful tool for analyzing complex, large-scale data, by capturing information in structured and expressive formats (see, e.g, \cite{billard_diday_2006,SDAoverview}). Interval-valued data, in particular, are useful for capturing inherent variability in numerical measurements. They arise naturally from data aggregation (macrodata) that summarize collections of individual observations (microdata). In this setting, the ontic perspective of intervals \citep{COUSO2014} is adopted, where an interval represents objective variability in an entity, treated as a realization of a set-valued variable. Common examples of interval-valued data include temperature ranges and price fluctuations. Such representations can also help address privacy concerns and enable efficient analysis of large datasets. A variety of methodologies have been developed for the analysis of interval-valued data, including extensions of clustering, regression, and principal component analysis \citep{DESOUZA2004,DESA2025,DIAS2017,Irpino2015regression,LeRademacher2012,RGSerrao2023}. In particular, several approaches have addressed outlier detection \citep{Li2006,Viattchenin2012,DuarteSilva2018,loureiro2026}, highlighting the growing interest in identifying atypical interval-valued observations.

Following the framework of \cite{Oliveira2022,oliveira2025}, interval-valued data are represented as tuples combining the center and range of the intervals (macrodata), and a distribution function characterizing the microdata. A key tool in this framework is the Mallows distance \citep{Mallows1972,Irpino2015}, which can be used to compare intervals. With this distance, a notion of central tendency, the barycenter, and a symbolic covariance matrix that captures scatter and association can be defined \citep{Irpino2015,oliveira2025}. These components are essential for the statistical analysis of interval-valued data. However, classical approaches are vulnerable to the presence of atypical observations, which can severely distort location and scale estimates. Robust methods offer a solution, with various robust methodologies having been proposed for interval-valued data \citep{FAGUNDES2013,DuarteSilva2018,Zhao2023,Tian2024,Greeshmagiri2025}. More specifically, \cite{loureiro2026} extends the Minimum Covariance Determinant (MCD) estimator \citep{Rousseeuw1984,Rousseeuw1985} to interval-valued data through the Interval Minimum Covariance Determinant (IMCD) estimator. The IMCD provides robust estimates of the barycenter and symbolic covariance matrix, enabling the definition of a robust Interval-Mahalanobis distance for effective outlier detection in this setting.

In this way, the robust Interval-Mahalanobis distance provides a robust measure of outlyingness. Nevertheless, robust detection alone does not address explainability, which remains challenging for interval-valued data. The Shapley value, a concept from cooperative game theory \citep{Shapley1953}, offers a promising framework for addressing this explainability challenge. Originally developed to fairly allocate contributions among players in a cooperative game, Shapley values have found new applications in ML for feature importance attribution \citep{SHAP2017}. In the context of outlier detection, \cite{shapley} demonstrated how Shapley values can be used to decompose the outlyingness of an observation into contributions from individual variables. This decomposition provides a clear and interpretable explanation of why a particular observation is considered an outlier, highlighting the specific features that contribute most to its outlying status.

In this work, we extend the Shapley value to the setting of interval-valued outlier detection. We derive a closed-form expression for the Shapley value of the squared robust Interval-Mahalanobis distance \cite{loureiro2026}, enabling efficient computation of variable contributions to outlyingness. We further provide a decomposition of contributions from centers, ranges, and cross-terms of the interval-valued observations, offering multiple levels of granularity in explaining outliers. This framework also connects naturally to cellwise outliers (see, e.g., \cite{Raymaekers_Rousseeuw_2021}), helping to identify outliers specific to certain variables that may not be atypical in the multivariate sense. Furthermore, we develop the Shapley interaction index, which reveals how pairs of variables jointly contribute to the outlyingness. Subsequently, we illustrate the practical utility of our approach through two real-world datasets. 

The structure of the paper is as follows. Section \ref{sec:intervaldata} introduces interval-valued data notation and framework. Section \ref{sec:outlier} reviews the outlier detection methodology based on the robust Interval-Mahalanobis distance. Section \ref{sec:explainability} develops the Shapley value methodology, including the Shapley interaction index. Section \ref{sec:examples} illustrates the methodology on two real-world datasets from automotive and music domains. Lastly, Section \ref{sec:conclusion} summarizes contributions. Proofs of theoretical results are provided in \ref{sec:appendix_explainability}.

\section{Interval-valued Data}
\label{sec:intervaldata}
In SDA, data assume more complex forms than traditional single-valued data, of which interval-valued data are a prominent example. Following the formulation in \cite{Oliveira2022,oliveira2025}, $\mathbb{IR} = \{[a,b]:\; a,b \in \mathbb{R},\; a \leq b \}$ denotes the set of all real closed and bounded intervals, and $\mathbb{IR}^p = \{([a_1,b_1],\dots,[a_p,b_p])^\top:\; a_j,b_j \in \mathbb{R},\; a_j \leq b_j, j=1,\dots,p \}$ represents the set of $p$-dimensional hyperrectangles, with $p$ $\in \mathbb{N}$. In addition, let $c=(a+b)/2\in\mathbb{R}$ and $r=b-a\in\mathbb{R}_0^+$ be the center and range of a real-valued interval $[a,b]\in\mathbb{IR}$, respectively. Then, the interval $[a,b]$ can be equivalently represented as $[c-r/2,c+r/2]$. 

In this framework, an observation can be interpreted as a \textit{symbolic interval}, which combines two complementary levels of information: the \textit{macrodata} which corresponds to the real-valued interval as defined by its bounds, and the \textit{microdata} which represents the underlying values within that range. Following \cite{oliveira2025}, we denote a symbolic interval by $\tilde{x}=(x,F_U)$, where $x=(c,r)^\top$ represents the real-valued interval corresponding to the macrodata and $F_U$ describes the distribution of the microdata normalized to the interval $[-1,1]$. Specifically, $F_U$ is the distribution function of a real-valued latent random variable, $U$, with support $[-1,1]$. If the microdata are not available, the distribution function $F_U$ is usually taken to be a continuous uniform distribution, which is the most common choice in SDA.

Let $\boldsymbol{\mathcal{\tilde{X}}} = (X_1,\ldots,X_p)^\top$ be an interval-valued random vector, and $\boldsymbol{\mathcal{C}} = (C_1,\ldots,C_p)^\top$ and $\boldsymbol{\mathcal{R}} = (R_1,\ldots,R_p)^\top$ be the real-valued random vectors of centers and ranges, respectively, such that $\mathbb{P}(\boldsymbol{\mathcal{R}} \geq \boldsymbol{0}) = 1$. Additionally, consider $\boldsymbol{U} = (U_1,\ldots,U_p)^\top$ as the real-valued latent random vector of independent and absolutely continuous random variables with support $\left[-1,1\right]^p$ and joint distribution function $F_{\boldsymbol{U}}$. Then, the interval-valued random vector $\boldsymbol{\mathcal{X}}$ can be characterized as $\boldsymbol{\mathcal{\tilde{X}}}=(\boldsymbol{\mathcal{X}},F_{\boldsymbol{U}})$, with $\boldsymbol{\mathcal{X}}=(\boldsymbol{\mathcal{C}}^\top,\boldsymbol{\mathcal{R}}^\top)^\top$. In particular, if $\mathbb{P}(\boldsymbol{\mathcal{R}} = \boldsymbol{0}) = 1$ and $\mathbb{P}(\boldsymbol{U} = \boldsymbol{0}) = 1$, then $\boldsymbol{\mathcal{\tilde{X}}}$ reduces to a conventional random vector, $\boldsymbol{\mathcal{X}}=\boldsymbol{\mathcal{C}}$. Finally, consider the real-valued random vector representing the microdata denoted by $\boldsymbol{V} = (V_1,\ldots,V_p)^\top\in\mathbb{R}^p$. Then, \cite{Oliveira2022} define $V_j=C_j+U_jR_j/2$, if $\mathbb{P}(R_j=0)=0$, and $V_j=C_j$, if $\mathbb{P}(R_j=0)=1$ and $\mathbb{P}(U_j=0)=1$ as the model linking the microdata and macrodata.

A natural way to compare intervals is the Mallows distance \citep{Mallows1972,Irpino2015}, which is a widely used metric in SDA. For interval-valued data, the Mallows distance is particularly relevant as it captures the unique structure of intervals, including its variability and microdata distribution. Let $\tilde{x}_1=((c_1,r_1)^\top,F_{1})$ and $\tilde{x}_2=((c_2,r_2)^\top,F_{2})$ be two univariate symbolic intervals, with $F_j$ and $F_{j}^{-1}$ the distribution function (with finite second moment) and quantile function of the microdata, respectively, for $j=1,2$. Then, the Mallows distance between $\tilde{x}_1$ and $\tilde{x}_2$ is defined as
\begin{equation}
    \label{eq:mallows_def}
    d_\mathrm{M}(\tilde{x}_1, \tilde{x}_2) = \left[\int_0^1 (F_{1}^{-1}(t) - F_{2}^{-1}(t))^2 \, dt\right]^\frac{1}{2}.
\end{equation}
This distance can be interpreted as the $L^2$ distance between the quantile functions of the microdata distributions.

Consider, now, that $\boldsymbol{\tilde{x}}_i=(\boldsymbol{x}_i,F_{\boldsymbol{U}})$ is the $i$-th realization of an interval-valued random vector $\boldsymbol{\mathcal{\tilde{X}}}$, with macrodata $\boldsymbol{x}_i=(\boldsymbol{c}_i^\top,\boldsymbol{r}_i^\top)^\top$, centers $\boldsymbol{c}_i=(c_{i1},\dots,c_{ip})^\top\in\mathbb{R}^p$, and ranges $\boldsymbol{r}_i=(r_{i1},\dots,r_{ip})^\top\in(\mathbb{R}^{+}_0)^p$, for $i=1,2$. Furthermore, since $\boldsymbol{U}$ has finite second moments, then, according to \cite{oliveira2025}, the Mallows distance between these intervals can be expressed as
\begin{equation}
    \label{eq:mallows}
    d_\mathrm{M}(\boldsymbol{x}_1, \boldsymbol{x}_2;F_{\boldsymbol{U}}) = \left[(\boldsymbol{c}_1-\boldsymbol{c}_2)^\top(\boldsymbol{c}_1-\boldsymbol{c}_2) + (\boldsymbol{r}_1-\boldsymbol{r}_2)^\top\boldsymbol{\Delta}(\boldsymbol{r}_1-\boldsymbol{r}_2) + (\boldsymbol{c}_1-\boldsymbol{c}_2)^\top\boldsymbol{\Psi}(\boldsymbol{r}_1-\boldsymbol{r}_2)\right]^\frac{1}{2},
\end{equation}
where $\boldsymbol{\Delta}=\mathrm{diag}(\delta_1,\dots,\delta_p)$, with $\delta_j=\mathbb{E}(U_j^2)/4$, $\boldsymbol{\Psi}=\mathrm{diag}(\mathbb{E}(U_1),\dots,\mathbb{E}(U_p))$, and $\mathrm{diag}(\boldsymbol{v})$ denotes the diagonal matrix whose main diagonal is the vector $\boldsymbol{v}$.

The Mallows distance is often used in SDA as a basis for some statistical concepts and methods. For instance, it can be used to define a notion of central tendency, which can serve as a representative of the data. The barycenter serves this purpose as the interval that minimizes the weighted sum of squared Mallows distances to a set of intervals \citep{Irpino2015,Irpino2006}. In fact, the population barycenter can be derived for any distribution of the microdata based on the model linking microdata and macrodata. Assuming that the centers and ranges, $\boldsymbol{\mathcal{C}}$ and $\boldsymbol{\mathcal{R}}$, have finite expected values, $\boldsymbol{\mu}_C$ and $\boldsymbol{\mu}_R$, respectively, and that the latent random vector $\boldsymbol{U}$ is independent of $\boldsymbol{\mathcal{X}}$, \cite{oliveira2025} showed that the population barycenter of $\boldsymbol{\mathcal{\tilde{X}}}$ is given by the interval
\begin{equation}
    \label{eq:barycenter}
    \boldsymbol{\mu}_B=(\boldsymbol{\mu}_C,\boldsymbol{\mu}_R,F_{\boldsymbol{U}}).
\end{equation}

Consider further that the centers and ranges, $\boldsymbol{\mathcal{C}}$ and $\boldsymbol{\mathcal{R}}$, have covariance matrices, $\boldsymbol{\Sigma}_{CC}$ and $\boldsymbol{\Sigma}_{RR}$, respectively, and $\boldsymbol{\Sigma}_{CR}$ is the covariance matrix between $\boldsymbol{\mathcal{C}}$ and $\boldsymbol{\mathcal{R}}$. Then, based on the Mallows distance and the barycenter, \cite{oliveira2025} also derived the closed expression for the symbolic covariance matrix of $\boldsymbol{\mathcal{\tilde{X}}}$,
\begin{equation}
    \label{eq:cov_matrix}
    \boldsymbol{\Sigma}_B=\boldsymbol{\Sigma}_{CC}+\dfrac{1}{4}\boldsymbol{\mathfrak{E}}_{UU}\bullet\boldsymbol{\Sigma}_{RR}+\dfrac{1}{2}\boldsymbol{\Sigma}_{CR}\boldsymbol{\Psi}+\dfrac{1}{2}\boldsymbol{\Psi}\boldsymbol{\Sigma}_{RC},
\end{equation}
where $[\boldsymbol{\mathfrak{E}}_{UU}]_{j\ell}=\mathcal{E}(U_j,U_\ell)$, $j\neq \ell$, with $\mathcal{E}(U_j,U_\ell)=\mathbb{E}(F_{U_j}^{-1}(T) F_{U_\ell}^{-1}(T))=\int_0^1 F_{U_j}^{-1}(t) F_{U_\ell}^{-1}(t) \, dt$, $T\sim\mathrm{Unif}(0,1)$, $[\boldsymbol{\mathfrak{E}}_{UU}]_{jj}=\mathbb{E}(U_j^2)$, $\boldsymbol{\Psi}=\mathrm{diag}(\mathbb{E}(U_1),\dots,\mathbb{E}(U_p))$, $\bullet$ denotes the Schur or entrywise product of matrices, and $[\boldsymbol{A}]_{j\ell}$ the entry $(j,\ell)$ of matrix $\boldsymbol{A}$. This covariance matrix captures the variability of the interval-valued random vector, taking into account the centers, ranges, and microdata distributions. It serves as a crucial component in various statistical analyses, including outlier detection, as it provides a measure of scatter and association that is representative of the unique structure of interval-valued data.

\section{Outlier Detection based on Robust Interval-Mahalanobis Distance}
\label{sec:outlier}
The presence of atypical observations in real-life datasets can significantly influence statistical analyses, leading to biased estimates and misleading conclusions. In particular, the sample estimators of the barycenter and the symbolic covariance matrix are sensitive to outliers, which can distort the estimation of location and scatter measures. This sensitivity can compromise the effectiveness of outlier detection methods that rely on these estimates. The IMCD algorithm \citep{loureiro2026} provides robust estimates of the barycenter and symbolic covariance matrix, extending the FastMCD algorithm \citep{fastMCD} to interval-valued data. This algorithm first identifies a subset of observations that minimizes the determinant of the symbolic covariance matrix, thus reducing the influence of outliers on these estimates. To improve efficiency, the observations are, then, reweighted based on robust distances and a cutoff value to refine the estimates. This cutoff value can be, for instance, based on the farness concept \citep{classmap}. For a detailed description of the IMCD algorithm, the reader is referred to \cite{loureiro2026}.

Another key concept in Multivariate Analysis is the Mahalanobis distance, with the definition for the interval-valued equivalent, called the Interval-Mahalanobis distance \citep{loureiro2026}, being given in Definition~\ref{def:int_mah_dist}. This distance incorporates the unique structure of interval-valued data, including the centers, ranges, and microdata distributions.

\begin{definition}[Interval-Mahalanobis Distance \protect{\citep[Def.~4.1]{loureiro2026}}]
    \label{def:int_mah_dist}
    Let $\boldsymbol{\tilde{x}}=(\boldsymbol{x},F_{\boldsymbol{U}})$ with $\boldsymbol{x}=(\boldsymbol{c}^\top,\boldsymbol{r}^\top)^\top$, $\boldsymbol{c}=(c_1,\dots,c_p)^\top\in\mathbb{R}^p$, $\boldsymbol{r}=(r_1,\dots,r_p)^\top\in(\mathbb{R}^{+}_0)^p$, and $\boldsymbol{U}=(U_1,\dots,U_p)^\top$ a latent random vector of independent random variables supported on $\left[-1,1\right]^p$ with joint distribution function $F_{\boldsymbol{U}}$. The Interval-Mahalanobis distance between $\boldsymbol{\tilde{x}}$ and the barycenter $\boldsymbol{\mu}_B$, as in \eqref{eq:barycenter}, of a population with symbolic covariance matrix $\boldsymbol{\Sigma}_B$, as in \eqref{eq:cov_matrix}, is defined as
    \begin{equation}
    \label{eq:d2_mah}
    \begin{aligned}
        d_\mathrm{IMah}(\boldsymbol{x};F_{\boldsymbol{U}},\boldsymbol{\mu}_B,\boldsymbol{\Sigma}_B)&=\Bigg[(\boldsymbol{c}-\boldsymbol{\mu}_C)^\top\boldsymbol{\Sigma}_B^{-1}(\boldsymbol{c}-\boldsymbol{\mu}_C)+\frac{1}{4}(\boldsymbol{r}-\boldsymbol{\mu}_R)^\top\left(\boldsymbol{\mathfrak{E}}_{UU}\bullet\boldsymbol{\Sigma}_B^{-1}\right)(\boldsymbol{r}-\boldsymbol{\mu}_R)\\
        &\qquad +(\boldsymbol{c}-\boldsymbol{\mu}_C)^\top\boldsymbol{\Sigma}_B^{-1}\boldsymbol{\Psi}(\boldsymbol{r}-\boldsymbol{\mu}_R)\Bigg]^\frac{1}{2}.
    \end{aligned}
    \end{equation}
\end{definition}

In what follows, $d_\mathrm{IMah}(\boldsymbol{x};F_{\boldsymbol{U}},\boldsymbol{\mu}_B,\boldsymbol{\Sigma}_B)$ will be denoted by $d_\mathrm{IMah}(\boldsymbol{x})$, for simplicity. In addition, a robust version of the Interval-Mahalanobis distance can be obtained by using the robust estimates of the barycenter and symbolic covariance matrix provided by the IMCD algorithm. In that case, $d_\mathrm{IMah}(\boldsymbol{x})$ measures the distance between an observation $\boldsymbol{x}$ and the robust location of the data (as estimated by the IMCD algorithm). For that reason, the squared robust Interval-Mahalanobis distance provides a robust measure of outlyingness for interval-valued observations, allowing for effective outlier detection while accounting for the complexities of the data structure. In that way, an observation $\boldsymbol{x}$ at a large distance to the robust location of the data is considered an outlier, i.e., if $d_\mathrm{IMah}(\boldsymbol{x})$ exceeds a cutoff value $\tau$.

Since the robust Interval-Mahalanobis distance has an unknown distribution, \cite{loureiro2026} suggest non-parametric cutoff values based on the farness concept \citep{classmap}. This concept involves performing a series of transformations on the squared robust Interval-Mahalanobis distances to obtain a distribution of outlier scores. To start with, the squared distances are standardized using the median and the median absolute deviation (MAD). Then, the robust Yeo-Johnson transformation \citep{yeojohnson2000,Raymaekers2024} is applied to account for skewness in the distribution of the standardized distances. Next, to obtain an approximately standard Normal distribution, another standardization is performed using the median and MAD of the transformed distances. Finally, the cumulative distribution function of these transformed distances is computed to obtain outlier scores that can be interpreted as probabilities of being an outlier. The cutoff value can then be set, for example, at $0.95$, such that observations with an outlier score exceeding this threshold are flagged as outliers, corresponding to an approximate $5\%$ false alarm rate under the standard Normal assumption.

\section{Explainability with Shapley Values}
\label{sec:explainability}
With the growing interest and dependability in the field of AI, the need has arisen for methods that help humans understand the reasoning behind the decisions or predictions made by AI models. This has led to the development of XAI, which aims at making the decisions and inner workings of AI models more transparent and interpretable to humans. XAI helps fighting the ``black-box'' paradigm of many AI models, being particularly relevant in high-stakes applications such as healthcare, finance, and criminal justice. Some of the most widely used XAI methods include Local Interpretable Model-Agnostic Explanations (LIME), and Shapley values. For a more comprehensive overview of XAI methods, we refer the reader to \cite{XAI2020} and \cite{XAI2024}.

In this work, we focus on Shapley values, which were first introduced by \cite{Shapley1953} as a solution concept in cooperative game theory. They provide a way to fairly distribute the total value generated by a coalition (set) of players among the individual players based on their contributions. Specifically, a coalitional game $(P, v)$ consists of a set of players $P=\{1,2,\dots,p\}$ and a characteristic or value function $v$ that assigns the worth $v(S)\in\mathbb{R}$ to each coalition $S\subseteq P$ of players. The Shapley value for player $j$ quantifies its contribution to the total value generated by the coalition and is defined as 
\begin{equation}
    \label{eq:shapley_game}
    \phi_j(v)=\sum_{S\subseteq P\setminus\{j\}}\frac{|S|!(p-|S|-1)!}{p!}\left(v(S\cup\{j\})-v(S)\right),
\end{equation}
where $|\cdot|$ denotes set cardinality. In this way, $\phi_j(v)$ represents the average marginal contribution of player $j$ across all possible coalitions $S$ that do not include player $j$.

The Shapley value is, in fact, the unique solution that satisfies the desirable axioms \citep{Young1985}: efficiency, whereby the total value is distributed among all players, i.e., $\sum_{j=1}^t\phi_j(v)=v(P)$; symmetry, meaning that if two players contribute equally to all coalitions, they receive the same Shapley value, i.e., if $v(S\cup\{j\})=v(S\cup\{\ell\})$ for all coalitions $S$ not containing $j$ or $\ell$, then $\phi_j(v)=\phi_\ell(v)$; the dummy player property, that is if a player does not contribute to any coalition, its Shapley value is zero, i.e., if $v(S\cup\{j\})=v(S)$ for all coalitions $S$ not containing $j$, then $\phi_j(v)=0$; and monotonicity, i.e., if for any two games $(P, v_1)$ and $(P, v_2)$ and all coalitions $S\subseteq P$ we have $v_1(S\cup\{j\})-v_1(S)\leq v_2(S\cup\{j\})-v_2(S)$, then $\phi_j(v_1)\leq \phi_j(v_2)$ for all $j\in P$.

When applying this concept to ML, Shapley values quantify how much each variable contributes to the model's output, considering all possible combinations of variables. They have gained popularity since SHapley Additive exPlanations (SHAP) was introduced by \cite{SHAP2017}, making them a widely used tool in the XAI toolbox.

In the context of outlier detection, explainability becomes relevant in helping understand why a particular observation is flagged as an outlier and which variables contribute most to its outlyingness. Shapley values fairly attribute the overall outlyingness (as measured by the robust distance) to individual variables, providing insight into the role each variable plays in the detection of outliers. In particular, \cite{shapley} proposed a decomposition of the squared Mahalanobis distance of an observation into outlyingness contributions of each variable using Shapley values. Unlike SHAP, which can be computationally expensive, the Shapley value for the Mahalanobis distance was explicitly derived as a linear function, thus being of low computational cost. This approach is extended to interval-valued data in the next section, where we derive the Shapley values for the squared Interval-Mahalanobis distance given in Definition \ref{def:int_mah_dist}.

\subsection{Shapley Values for Interval-Valued Data}
\label{sec:shapley}
In this section, we follow the link between Shapley values and the Mahalanobis distance established by \cite{shapley} to extend the definition of the Shapley values to interval-valued data. Here, the set of players becomes the set of variables $P=\{1,\dots,p\}$, and a coalition is a subset $S\subseteq P$. Furthermore, the characteristic function $v$ is now $d_\mathrm{IMah}^2(\boldsymbol{x})$. Then, the Shapley value of a variable $j\in P$, in this context, is defined as the average marginal contribution of variable $j$ to $d_\mathrm{IMah}^2(\boldsymbol{x})$ across all subsets of variables that do not include $j$. This is reflected in the following theorem.

\begin{theorem}
\label{thm:shapley}
    Let $\boldsymbol{\tilde{x}}=(\boldsymbol{x},F_{\boldsymbol{U}})$ with $\boldsymbol{x}=(\boldsymbol{c}^\top,\boldsymbol{r}^\top)^\top$, $\boldsymbol{c}=(c_1,\dots,c_p)^\top\in\mathbb{R}^p$, $\boldsymbol{r}=(r_1,\dots,r_p)^\top\in(\mathbb{R}^{+}_0)^p$, and $\boldsymbol{U}=(U_1,\dots,U_p)^\top$ a random vector of independent random variables supported on $\left[-1,1\right]^p$ with distribution function $F_{\boldsymbol{U}}$.
    Let $\boldsymbol{\mu}_B$ be the barycenter, as in \eqref{eq:barycenter}, and $\boldsymbol{\Sigma}_B$ be the covariance matrix, as in \eqref{eq:cov_matrix}. For $j=1,\dots,p$, the Shapley value of variable $j$ is
        \begin{align}
            \label{eq:shapley}
            \phi_j(\boldsymbol{x};F_{\boldsymbol{U}},\boldsymbol{\mu}_B,\boldsymbol{\Sigma}_B)&=\sum_{S\subseteq P\setminus\{j\}}\frac{|S|!(p-|S|-1)!}{p!}\left(d_\mathrm{IMah}^2(\boldsymbol{x}^{S\cup \{j\}})-d_\mathrm{IMah}^2(\boldsymbol{x}^S)\right)\nonumber\\
            &=(c_j-\mu_{C_j})\sum_{\ell=1}^p(c_\ell-\mu_{C_\ell})\omega_{j\ell}+\frac{1}{4}(r_j-\mu_{R_j})\sum_{\ell=1}^p(r_\ell-\mu_{R_\ell})\mathcal{E}_{j\ell}\omega_{j\ell}\\
            &+\frac{1}{2}(c_j-\mu_{C_j})\sum_{\ell=1}^p(r_\ell-\mu_{R_\ell})\psi_{\ell\ell}\omega_{j\ell}+\frac{1}{2}(r_j-\mu_{R_j})\psi_{jj}\sum_{\ell=1}^p(c_\ell-\mu_{C_\ell})\omega_{j\ell}\nonumber,
        \end{align}
    where $\mathcal{E}_{j\ell}=[\boldsymbol{\mathfrak{E}}_{UU}]_{j\ell}$, $\psi_{j\ell}=[\boldsymbol{\Psi}]_{j\ell}$, $\omega_{j\ell}=[\boldsymbol{\Sigma}_B^{-1}]_{j\ell}$, $\ell=1,\dots,p$, and $\boldsymbol{x}^S=(x_1^S,\dots,x_p^S)^\top$ with $x_j^S=x_j=(c_j,r_j)^\top$, if $j\in S$, and $x_j^S=\mu_j=(\mu_{C_j},\mu_{R_j})^\top$, if $j\notin S$.
\end{theorem}
\begin{proof}
    The proof of this theorem is given in \ref{sec:proof_shapley}.
\end{proof}

The notation can once again be simplified by letting $\phi_j(\boldsymbol{x}):=\phi_j(\boldsymbol{x};F_{\boldsymbol{U}},\boldsymbol{\mu}_B,\boldsymbol{\Sigma}_B)$. The Shapley value can also be expressed in vector form, as shown in the following corollary.

\begin{corollary}
    Under the conditions of \autoref{thm:shapley}, the vector of Shapley values of an observation $\boldsymbol{x}$, $\boldsymbol{\phi}(\boldsymbol{x})=(\phi_1(\boldsymbol{x}),\dots,\phi_p(\boldsymbol{x}))^\top$, is given by
    \begin{equation}
    \label{eq:shapley_vector}
    \begin{aligned}
        \boldsymbol{\phi}(\boldsymbol{x})
        &=(\boldsymbol{c}-\boldsymbol{\mu}_C)\bullet\left[\boldsymbol{\Sigma}_B^{-1}(\boldsymbol{c}-\boldsymbol{\mu}_C)\right]+\frac{1}{4}(\boldsymbol{r}-\boldsymbol{\mu}_R)\bullet\left[\left(\boldsymbol{\mathfrak{E}}_{UU}\bullet\boldsymbol{\Sigma}_B^{-1}\right)(\boldsymbol{r}-\boldsymbol{\mu}_R)\right]\\
        &\quad+\frac{1}{2}(\boldsymbol{c}-\boldsymbol{\mu}_C)\bullet\left[\boldsymbol{\Sigma}_B^{-1}\boldsymbol{\Psi}(\boldsymbol{r}-\boldsymbol{\mu}_R)\right]+\frac{1}{2}(\boldsymbol{r}-\boldsymbol{\mu}_R)\bullet\left[\boldsymbol{\Psi}\boldsymbol{\Sigma}_B^{-1}(\boldsymbol{c}-\boldsymbol{\mu}_C)\right].
    \end{aligned}
    \end{equation}
\end{corollary}

Similarly to the conventional case, the Shapley value for the Interval-Mahalanobis distance is of linear complexity, as it can be computed in closed form without the need to evaluate the distance for all possible subsets of variables. This is a significant advantage over other XAI methods, such as SHAP, which can be computationally expensive due to the combinatorial nature of the problem.

In particular, we know from the Shapley value efficiency property that the sum of the Shapley values of all variables equals the squared Interval-Mahalanobis distance of the observation to the barycenter, i.e., $\sum_{j=1}^p\phi_j(\boldsymbol{x})=d_\mathrm{IMah}^2(\boldsymbol{x})$. This property ensures that the contributions of all variables sum up to the total outlyingness of the observation, as measured by the robust distance. In this way, if $\phi_j(\boldsymbol{x})$ is large, it means that variable $j$ contributes significantly to the outlyingness of observation $\boldsymbol{x}$. We also note that negative Shapley values are possible, indicating that variable $j$ has a stabilizing effect on the outlyingness of the observation. More specifically, if $\phi_j(\boldsymbol{x})<0$, it means that replacing the value of variable $j$ in the observation $\boldsymbol{x}$ with its barycenter value $\mu_j$ would lead to an average increase in $d_\mathrm{IMah}^2(\boldsymbol{x})$. In other words, the variable $j$ is mitigating the outlyingness of the observation $\boldsymbol{x}$, i.e., it is pulling the observation towards the barycenter. This is particularly useful in the context of outlier detection, as it allows identifying which variables are driving the outlyingness and which ones are weakening it.

The Shapley value provides essentially an additive decomposition of the squared Interval-Mahalanobis distance. Therefore, looking at the expression in \eqref{eq:shapley_vector}, we can also decompose the Shapley value into contributions of the centers, $(\boldsymbol{c}-\boldsymbol{\mu}_C)\bullet\left[\boldsymbol{\Sigma}_B^{-1}(\boldsymbol{c}-\boldsymbol{\mu}_C)\right]$, the ranges, $1/4(\boldsymbol{r}-\boldsymbol{\mu}_R)\bullet\left[\left(\boldsymbol{\mathfrak{E}}_{UU}\bullet\boldsymbol{\Sigma}_B^{-1}\right)(\boldsymbol{r}-\boldsymbol{\mu}_R)\right]$, and the cross centers-ranges, $1/2(\boldsymbol{c}-\boldsymbol{\mu}_C)\bullet\left[\boldsymbol{\Sigma}_B^{-1}\boldsymbol{\Psi}(\boldsymbol{r}-\boldsymbol{\mu}_R)\right]+1/2(\boldsymbol{r}-\boldsymbol{\mu}_R)\bullet\left[\boldsymbol{\Psi}\boldsymbol{\Sigma}_B^{-1}(\boldsymbol{c}-\boldsymbol{\mu}_C)\right]$. This furthers outlyingness explainability in the context of interval-valued data, adding a new level to the analysis.

Another interesting aspect of the Shapley value is its connection to the concept of cellwise outliers (see, e.g., \cite{Raymaekers_Rousseeuw_2021}), which are observations that are not outlying in their entirety but contain specific cells (variable values) that are outlying. When considering the Shapley values per variable, we can identify which observations behave as outliers in specific variables, even if they are not outliers in the overall distance. Thus, the Shapley value provides a more granular understanding of outliers, allowing us to not only identify which variables are driving the outlyingness of an observation but also to flag potential cellwise outliers that may be masked when looking at the overall distance.

\subsection{Shapley Interaction Index}
\label{sec:shapley_interaction}
A natural extension of the Shapley value is the Shapley interaction index, another concept from cooperative game theory \citep{Grabisch1999,Fujimoto2006}. The Shapley interaction index for $Q$ is defined as
\begin{equation}
    \label{eq:shapley_interaction_game}
    \Phi_Q(v)=\sum_{T\subseteq P\setminus Q}\frac{|T|!(p-|T|-|Q|)!}{(p-|Q|+1)!}\Delta_Q\, v(T),
\end{equation}
where $\Delta_Q v(T)=\sum_{S\subseteq Q}(-1)^{|Q|-|S|}v(T\cup S)$. The Shapley interaction index $\Phi_Q(v)$ quantifies the interaction effect of the coalition $Q$ on the total value generated by the coalition, taking into account all possible coalitions that do not include $Q$. This concept has gained popularity in the context of XAI as SHAP Interaction Values \citep{SHAP2018}, which account for interactions between variables. In this way, we can obtain information not only about the individual contributions of each variable but also about their joint effect.

When considering outlier detection, the Shapley interaction index quantifies the contribution of a pair of variables to the outlyingness of an observation, taking into account their joint effect. This is particularly relevant in multivariate data, where variables can interact in complex ways. To address this, \cite{shapley} generalized their proposal of Shapley value to the Shapley interaction index. We extend this concept to interval-valued data, with \autoref{thm:shapley_interaction}, where $Q$ is now a pair of variables $(j,\ell)$, and the Shapley interaction index $\Phi_{j\ell}(v)$ measures the pairwise contribution of variables $j$ and $\ell$ to $d_\mathrm{IMah}^2(\boldsymbol{x})$.

\begin{theorem}
\label{thm:shapley_interaction}
    Under the conditions of \autoref{thm:shapley}, let $T\subseteq P\setminus\{j,\ell\}$ be a subset of variables that does not include the variable pair $(j,\ell)$.
    The Shapley interaction indexes for the pairs $(j,\ell)$, $j,\ell=1,\dots,p$, are collected in the $p\times p$ matrix $\boldsymbol{\Phi}(\boldsymbol{x})=\boldsymbol{\Phi}(\boldsymbol{x}; F_{\boldsymbol{U}}, \boldsymbol{\mu}_B,\boldsymbol{\Sigma}_B)$, where $[\boldsymbol{\Phi}]_{j\ell}=\Phi_{j\ell}$. Then, the off-diagonal elements are
    \begin{equation}
    \label{eq:shapley_interaction}
    \begin{aligned}
        \Phi_{j\ell}(\boldsymbol{x})&=\sum_{T\subseteq P\setminus\{j,\ell\}}\frac{|T|!(p-|T|-2)!}{(p-1)!}\Delta_{\{j,\ell\}}d_\mathrm{IMah}^2(\boldsymbol{x}^T)\\
        &=2(c_j-\mu_{C_j})(c_\ell-\mu_{C_\ell})\omega_{j\ell}+\frac{1}{2}(r_j-\mu_{R_j})(r_\ell-\mu_{R_\ell})\mathcal{E}_{j\ell}\omega_{j\ell}\\
        &\quad+(c_j-\mu_{C_j})(r_\ell-\mu_{R_\ell})\psi_{\ell\ell}\omega_{j\ell}+(r_j-\mu_{R_j})(c_\ell-\mu_{C_\ell})\psi_{jj}\omega_{j\ell},
    \end{aligned}
    \end{equation}
    with $\Delta_{\{j,\ell\}}d_\mathrm{IMah}^2(\boldsymbol{x}^T)=d_\mathrm{IMah}^2(\boldsymbol{x}^{T\cup\{j,\ell\}})-d_\mathrm{IMah}^2(\boldsymbol{x}^{T\cup\{j\}})-d_\mathrm{IMah}^2(\boldsymbol{x}^{T\cup\{\ell\}})+d_\mathrm{IMah}^2(\boldsymbol{x}^T)$,
    and the diagonal elements are
    \begin{equation}
    \label{eq:shapley_interaction_diag}
    \begin{aligned}
        \Phi_{jj}(\boldsymbol{x})&=\phi_j(\boldsymbol{x})-\sum_{\ell\neq j}\Phi_{j\ell}(\boldsymbol{x})\\
        &=(c_j-\mu_{C_j})^2\omega_{jj}-(c_j-\mu_{C_j})\sum_{\ell\neq j}(c_\ell-\mu_{C_\ell})\omega_{j\ell}\\
        &\quad+\frac{1}{4}(r_j-\mu_{R_j})^2\mathcal{E}_{jj}\omega_{jj}-\frac{1}{4}(r_j-\mu_{R_j})\sum_{\ell\neq j}(r_\ell-\mu_{R_\ell})\mathcal{E}_{j\ell}\omega_{j\ell}\\
        &\quad+\frac{1}{2}(c_j-\mu_{C_j})(r_j-\mu_{R_j})\psi_{jj}\omega_{jj}-\frac{1}{2}(c_j-\mu_{C_j})\sum_{\ell\neq j}(r_\ell-\mu_{R_\ell})\psi_{\ell\ell}\omega_{j\ell}\\
        &\quad+\frac{1}{2}(r_j-\mu_{R_j})(c_j-\mu_{C_j})\psi_{jj}\omega_{jj}-\frac{1}{2}(r_j-\mu_{R_j})\psi_{jj}\sum_{\ell\neq j}(c_\ell-\mu_{C_\ell})\omega_{j\ell}.
    \end{aligned}
    \end{equation}
\end{theorem}

\begin{proof}
    The proof of this theorem is given in \ref{sec:proof_shapley_interaction}.
\end{proof}

We note that the diagonal elements $\Phi_{jj}(\boldsymbol{x})$ are defined so that each Shapley value can be decomposed as $\phi_j(\boldsymbol{x})=\sum_{\ell=1}^p \Phi_{j\ell}(\boldsymbol{x})$. Consequently, the total outlyingness satisfies $d_\mathrm{IMah}^2(\boldsymbol{x})=\sum_{j=1}^p\sum_{\ell=1}^p \Phi_{j\ell}(\boldsymbol{x})$, thereby extending the efficiency property of the Shapley value to the Shapley interaction index. Therefore, the off-diagonal elements $\Phi_{j\ell}(\boldsymbol{x})$ represent the joint contribution of variables $j$ and $\ell$, while the diagonal elements $\Phi_{jj}(\boldsymbol{x})$ represent the individual contribution of variable $j$ after accounting for its interactions with all other variables. The Shapley interaction index matrix can also be expressed in a more compact form, leading to the following corollary.

\begin{corollary}
    Under the conditions of \autoref{thm:shapley_interaction}, the Shapley interaction index of an interval-valued observation $\boldsymbol{x}$ is given by the matrix
    \begin{equation}
        \begin{aligned}
            \boldsymbol{\Phi}(\boldsymbol{x})&=2(\boldsymbol{c}-\boldsymbol{\mu}_C)(\boldsymbol{c}-\boldsymbol{\mu}_C)^\top\bullet\boldsymbol{\Sigma}_B^{-1}+\frac{1}{2}(\boldsymbol{r}-\boldsymbol{\mu}_R)(\boldsymbol{r}-\boldsymbol{\mu}_R)^\top\bullet\boldsymbol{\mathfrak{E}}_{UU}\bullet\boldsymbol{\Sigma}_B^{-1}+\\
            &\quad+(\boldsymbol{c}-\boldsymbol{\mu}_C)(\boldsymbol{r}-\boldsymbol{\mu}_R)^\top\bullet\boldsymbol{\Sigma}_B^{-1}\boldsymbol{\Psi}+(\boldsymbol{r}-\boldsymbol{\mu}_R)(\boldsymbol{c}-\boldsymbol{\mu}_C)^\top\bullet\boldsymbol{\Psi}\boldsymbol{\Sigma}_B^{-1}-\mathrm{diag}(\boldsymbol{\phi}(\boldsymbol{x})).
        \end{aligned}
    \end{equation}
\end{corollary}

\section{Examples}
\label{sec:examples}
We revisit the two case studies introduced in \cite{loureiro2026}, retaining the same data construction, preprocessing pipeline, and outlier detection method. The present work builds upon these results by incorporating the proposed Shapley value-based framework for outlier interpretability.

\subsection{Cars Dataset}
\label{sec:cars}
The Cars dataset is originally available from the \texttt{MAINT.Data} R package \citep{MAINT.Data}. It is composed of $27$ observations corresponding to different car models, with four interval-valued variables: \textit{Price}, \textit{Engine Capacity} (\textit{EngCap}), \textit{Top Speed}, and \textit{Acceleration}; and a class variable with values: \textit{Berlina}, \textit{Luxury}, \textit{Sportive}, and \textit{Utilitarian}. We follow the same methodology as in \cite{loureiro2026}, applying a logarithmic transformation to the \textit{Price} variable (denoted as \textit{lnPrice}), and, since no microdata are available, assuming a continuous uniform distribution for the latent variables. As for outlier detection, we computed the squared robust Interval-Mahalanobis distances of the observations to the barycenter and used the farness concept with a cutoff of $0.9$ to flag outliers, as described in Section \ref{sec:outlier}. In particular, the sample barycenter and covariance matrix were estimated using the reweighted IMCD algorithm with a subset of approximately $75\%$ of the total observations and a farness cutoff of $0.9$.

\autoref{fig:cars_dist} shows the squared robust Interval-Mahalanobis distances, ordered from largest to smallest, and the $0.9$ farness cutoff value, with the outliers marked in blue. We see that five observations were identified as outliers: \textit{MercedesClasseS}, \textit{Ferrari}, \textit{Porsche}, \textit{MercedesSL}, and \textit{HondaNSK}. In addition, the class variable is represented by the shape of the points, and we can see that the outliers are all from the \textit{Luxury} class, with only \textit{MercedesClasseS} being from the \textit{Sportive} class.

\begin{figure}[ht]
    \centering
    \includegraphics[width=\textwidth]{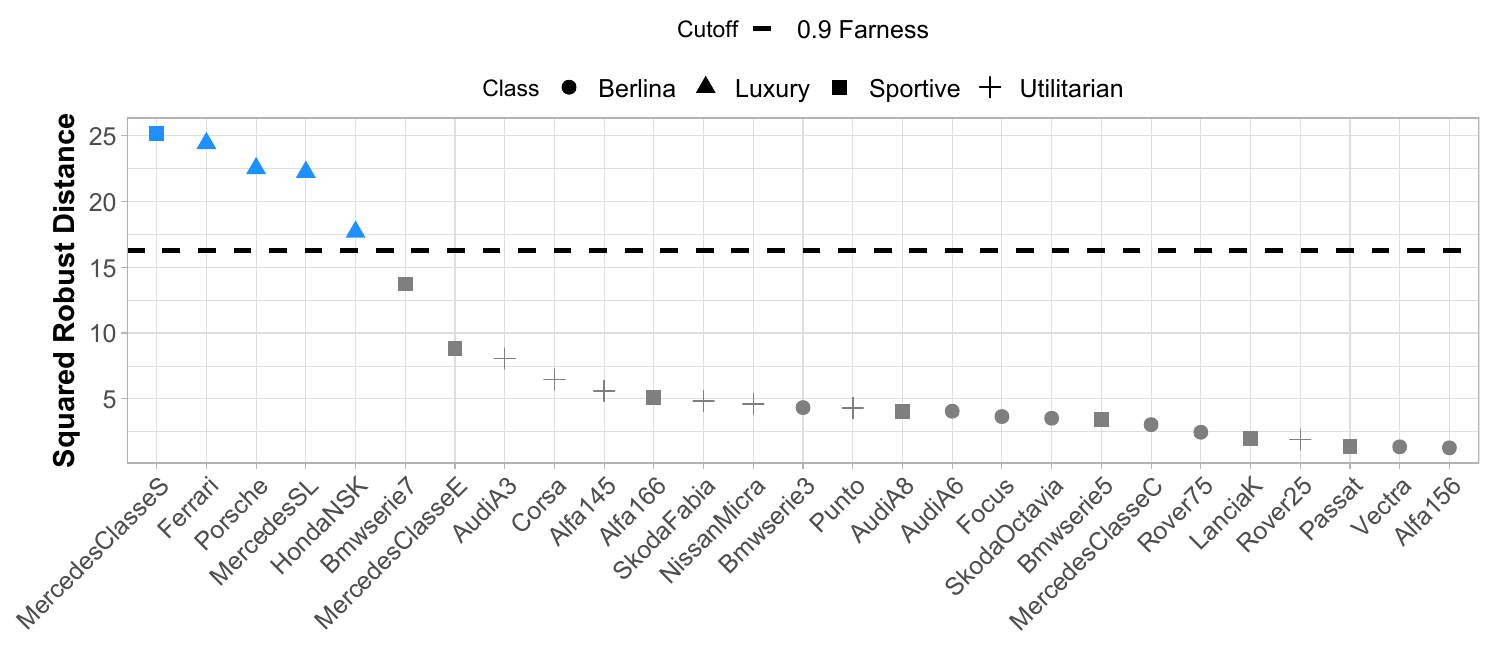}
    \caption{Plot of the squared robust Interval-Mahalanobis distances for the Cars dataset. The points shape represents the car model class, with the horizontal line corresponding to the $0.9$ farness cutoff value, and the outliers being marked in blue.}
    \label{fig:cars_dist}
\end{figure}

To further explore the reasons behind the outlyingness of these observations, we computed the Shapley values for the squared robust Interval-Mahalanobis distances. A barplot of the Shapley value decomposition for the six observations with the largest distances is presented in \autoref{fig:cars_shapley_barplot}. The black vertical lines represent the squared distances, while the dashed horizontal line represents the $0.9$ farness cutoff. We see that for \textit{MercedesClasseS}, \textit{MercedesSL}, and \textit{Bmwserie7} the variable with the largest Shapley value is the \textit{EngCap}, while for \textit{Ferrari} and \textit{Porsche} it is the \textit{Top Speed}. For \textit{Ferrari}, \textit{EngCap} has a Shapley value almost as large as \textit{Top Speed}. Moreover, for \textit{HondaNSK} the largest Shapley value is from \textit{lnPrice}, which also plays a part in the \textit{MercedesClasseS} case. Finally, we see that the \textit{lnPrice} plays a stabilizing role for the \textit{Ferrari}, \textit{Porsche}, and \textit{MercedesSL}. It is clear that the method is identifying car models with atypical engine capacities, top speeds, or a combination of these factors, as outliers. This is concurrent with the idea that luxury or sports cars are usually more powerful and faster than regular cars, standing apart from them.

\begin{figure}[ht]
    \centering
    \begin{subfigure}[t]{0.49\textwidth}
        \centering
        \includegraphics[width=\textwidth]{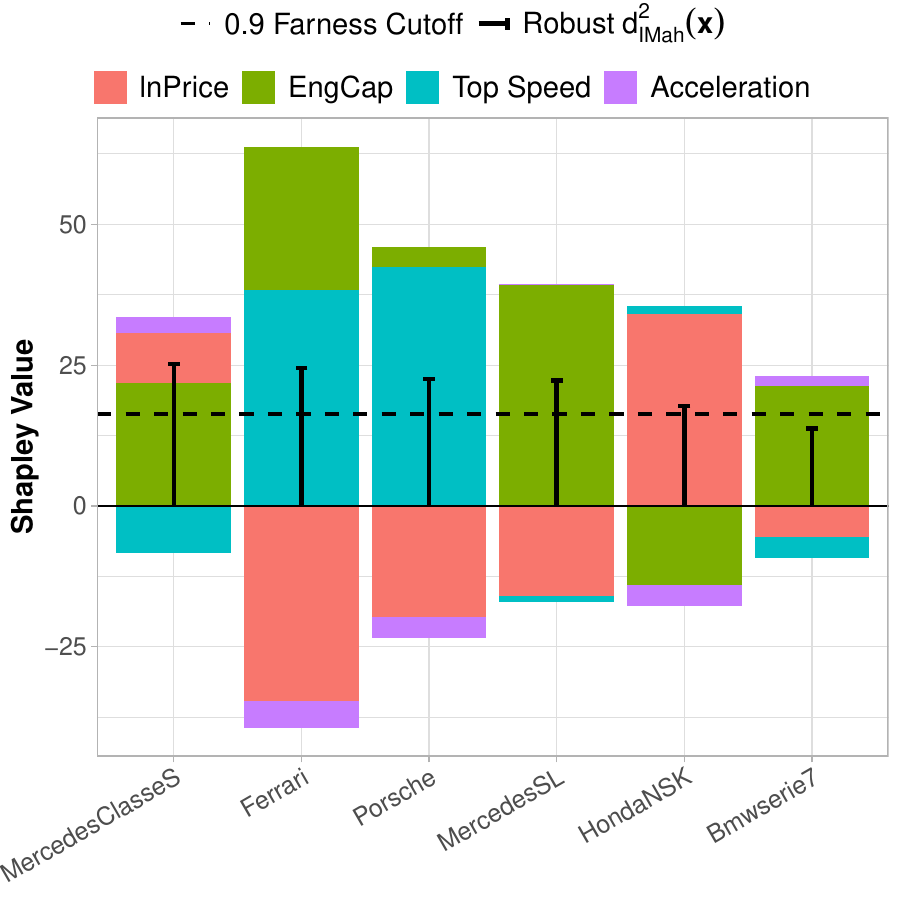}
        \caption{Barplot of the Shapley values.}
        \label{fig:cars_shapley_barplot}
    \end{subfigure}
    \hfill
    \begin{subfigure}[t]{0.49\textwidth}
        \centering
        \includegraphics[width=\textwidth]{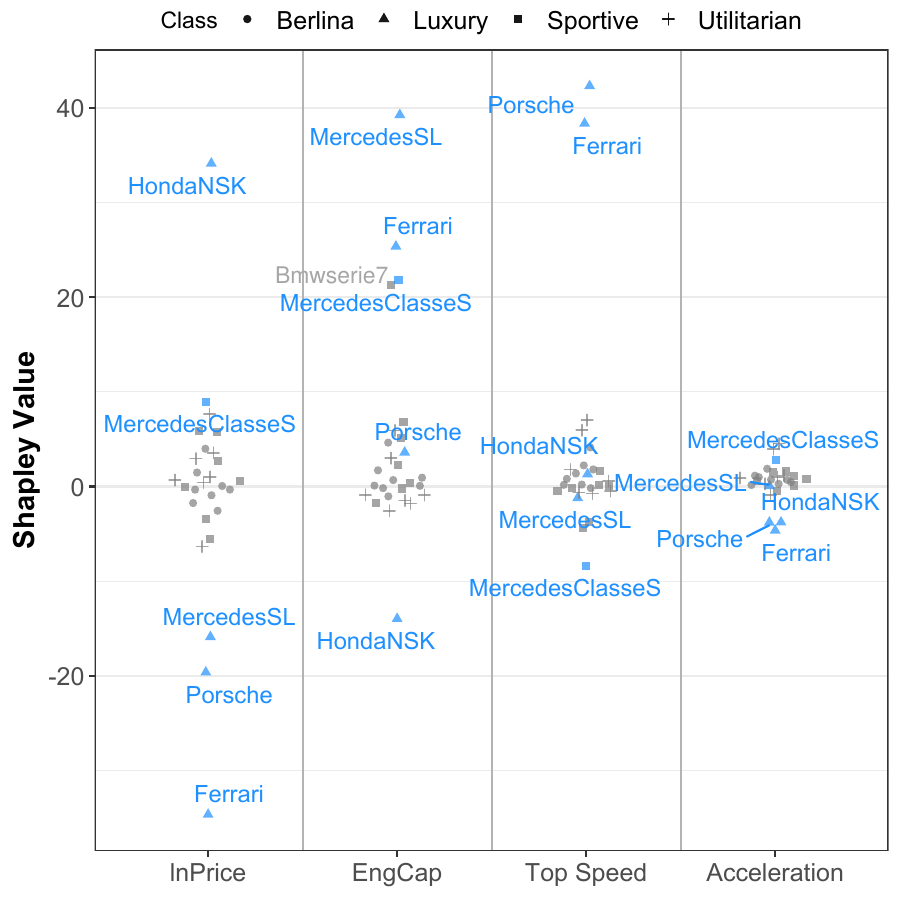}
        \caption{Beeswarm plot of the Shapley values.}
        \label{fig:cars_shapley_beeswarm}
    \end{subfigure}
    \caption{Barplot (a) and beeswarm plot (b) of the Shapley value decomposition of the squared robust Interval-Mahalanobis distances for the Cars dataset. In the barplot (a), only the six observations with the largest distances are shown, ordered from greatest to smallest distance represented by the black vertical lines, with the dashed horizontal line representing the $0.9$ farness cutoff. In the beeswarm plot (b), the outlying observations are marked in blue and the points shape represents the car model class.}
    \label{fig:cars_shapley}
\end{figure}

\autoref{fig:cars_shapley_beeswarm} in turn shows a beeswarm plot of the Shapley values for all observations per variable. Here we can see how the outliers compare to the regular observations. For instance, the \textit{Acceleration} variable is the one with more concentrated Shapley values, and, essentially, less contribution to the distances. \textit{EngCap} and \textit{TopSpeed}, in general, seem to have the greatest contributions to the distances, while the \textit{lnPrice}, for some cases, markedly reduces the outlyingness, with \textit{EngCap} also substantially reducing the outlyingness of \textit{HondaNSK}. The Shapley values when looked at from this perspective are closely related to the concept of cellwise outliers. Here, \textit{Bmwserie7} serves as an example of a potential cellwise outlier, with a high outlying value in the \textit{EngCap} variable but not being flagged as an outlier. This is due to the stabilizing effect of \textit{lnPrice} and \textit{TopSpeed}, as their Shapley values are negative, as shown in \autoref{fig:cars_shapley_barplot}.

As mentioned in Section \ref{sec:shapley}, the Shapley value can be decomposed into contributions of the centers, ranges, and cross centers-ranges. Since we assumed symmetric and identically distributed latent variables, the cross centers-ranges contributions are zero, and the Shapley value can be decomposed into contributions of the centers and ranges only. This decomposition is shown in \autoref{fig:cars_shapley_decomp} for the six observations with the largest distances. We see that for all the observations, the contribution of the centers is considerably larger than that of the ranges, which is expected since the expression for the distance, and consequently for the Shapley value, limits the ranges' weight. This suggests that the outlyingness of these observations is mainly driven by their centers, rather than by their ranges.

\begin{figure}[ht]
    \centering
    \begin{subfigure}[t]{0.64\textwidth}
        \centering
        \includegraphics[width=\textwidth]{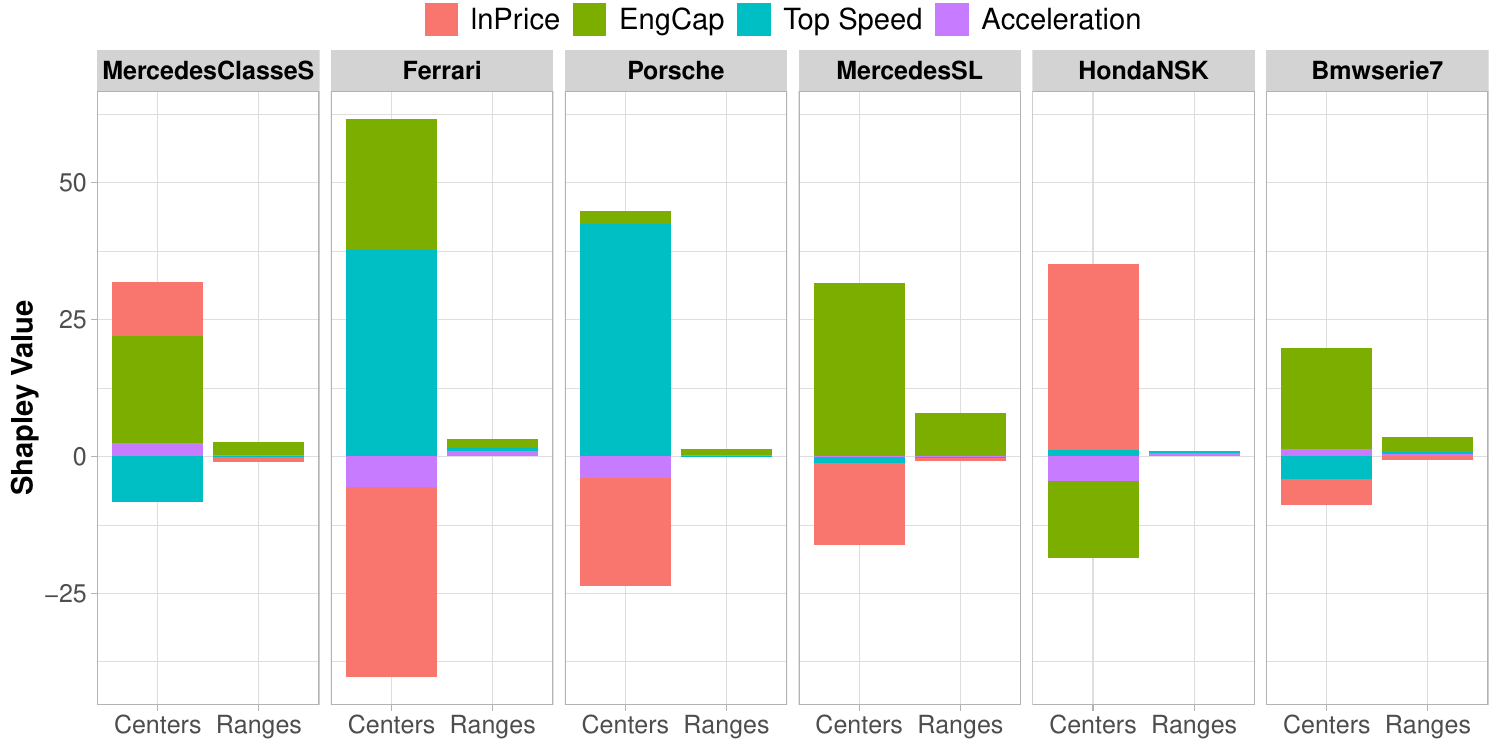}
        \caption{Barplot of the decomposition of the Shapley values by centers and ranges.}
        \label{fig:cars_shapley_decomp}
    \end{subfigure}
    \hfill
    \begin{subfigure}[t]{0.34\textwidth}
        \centering
        \includegraphics[width=\textwidth]{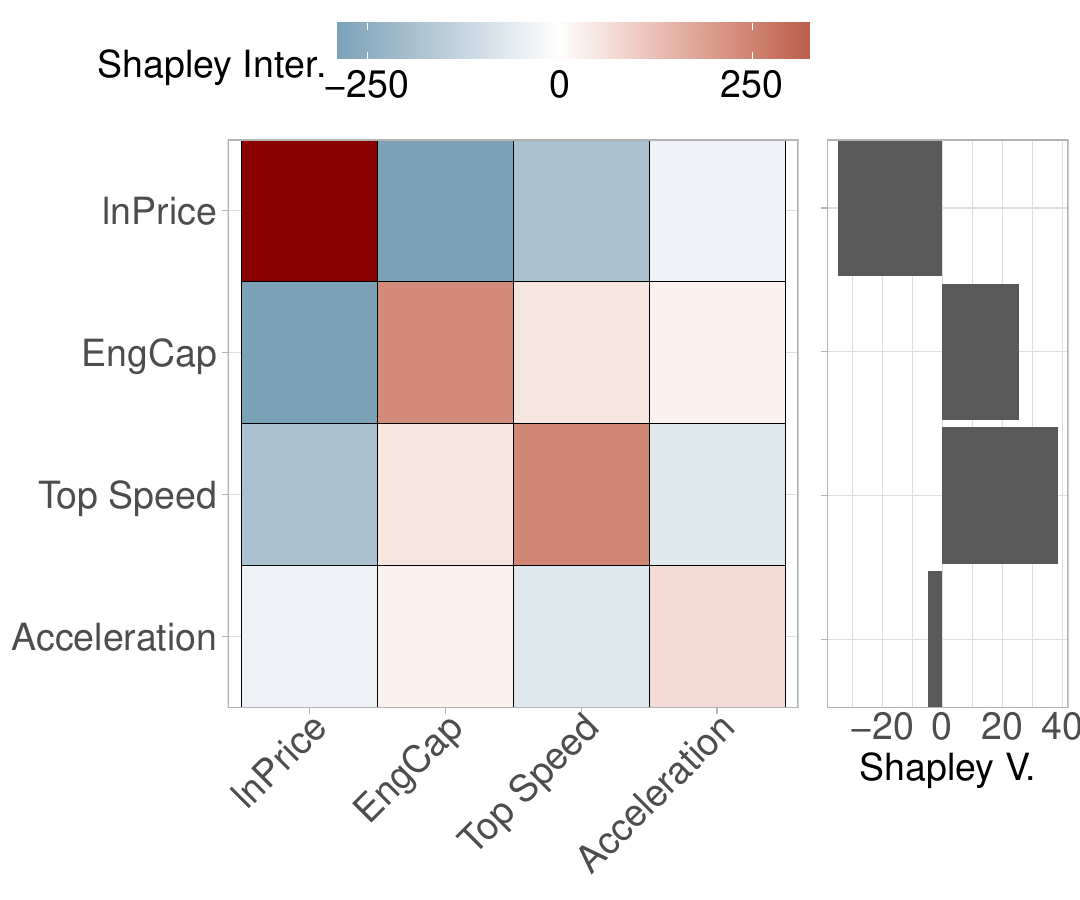}
        \caption{Tile plot of the Shapley interaction index for \textit{Ferrari}.}
        \label{fig:cars_shapley_inter}
    \end{subfigure}
    \caption{Barplot (a) of the decomposition of the Shapley values by centers and ranges for the six observations with the largest distances in the Cars dataset. Tile plot (b) of the Shapley interaction index for \textit{Ferrari}, with the corresponding Shapley values on the right.}
    \label{fig:cars2}
\end{figure}

Lastly, we computed the Shapley interaction index, which allows us to further explore the joint contribution of pairs of variables to the outlyingness of an observation. For instance, \autoref{fig:cars_shapley_inter} shows a tile plot of the Shapley interaction index for the \textit{Ferrari} observation. We see that the largest positive values are in the diagonal, which represent the individual contributions of each variable, while the off-diagonal elements represent the joint contributions of pairs of variables. In particular, the individual contribution of \textit{lnPrice} is the largest positive contribution, but the negative joint contributions with the other variables are also significant, and combined they actually have a stabilizing effect on the outlyingness of \textit{Ferrari}, indicating that the atypical price of \textit{Ferrari} is somewhat mitigated by its engine capacity, top speed, and acceleration. In contrast, \textit{EngCap} has positive joint contributions with \textit{Top Speed} and \textit{Acceleration}, which further increase the outlyingness of \textit{Ferrari}. This example illustrates how the Shapley interaction index can provide insights into the outlyingness of an observation, which may not be apparent when looking at marginal contributions alone.

\subsection{Spotify Dataset}
The Spotify Tracks dataset \citep{kaggle.spotify2022}, collected in October 2022, comprises $20$ audio and metadata features describing $114\,000$ music tracks. In this case study, tracks are aggregated by genre to detect outlying genres, and Shapley values are used to assess the contribution of each audio feature to their outlyingness. 

For this dataset, duplicates were removed, and the analysis was restricted to the $11$ numerical variables: \textit{duration\_ms}, \textit{popularity}, \textit{danceability}, \textit{energy}, \textit{loudness}, \textit{speechiness}, \textit{acousticness}, \textit{instrumentalness}, \textit{liveness}, \textit{valence}, and \textit{tempo}. To handle tracks associated with multiple genres, robust centroids were estimated for each genre based on the single-genre tracks, and, then, multi-genre tracks were assigned to the closest genre according to the robust Mahalanobis distance. This process yielded a dataset of $81\,033$ tracks. Several transformations were subsequently applied, including logarithmic scaling for \textit{tempo} and \textit{loudness}, normalization of \textit{popularity} to $[0,1]$, and conversion of \textit{duration\_ms} into minutes (\textit{duration}). Then, the data aggregation was carried out by setting the $1\%$ and $99\%$ quantiles as the lower and upper bounds of the intervals, respectively. The final dataset consists of $111$ genres and $11$ audio features, with each interval-valued observation representing a median of $800$ tracks. See \cite{loureiro2026} for further details.

Since, in this case, the microdata are available, we estimated the parameters of the latent variables' distributions, using Kernel Density Estimation (KDE) and the R-package \texttt{kde1d} \citep{kde1d}. The squared robust Interval-Mahalanobis distances of the observations to the barycenter were computed using the reweighted IMCD estimator with a subset size of approximately $75\%$ of the total number of observations and a farness cutoff of $0.95$. The outliers were then flagged using the farness concept with a cutoff of $0.95$ for extreme outliers and $0.9$ for mild outliers. In particular, \autoref{fig:spotify_dist} shows the squared robust distances and the cutoff values, with the extreme and mild outliers highlighted in blue and green, respectively. We see that \textit{sleep}, \textit{comedy}, \textit{classical}, and \textit{grindcore} are identified as extreme outliers, and \textit{chicago-house}, \textit{ambient}, \textit{iranian}, and \textit{death-metal} as mild outliers. Basic musical knowledge can help us understand why these genres are being flagged as outliers, and the Shapley values can further clarify which features are contributing to their outlyingness.

\begin{figure}[ht]
    \centering
    \includegraphics[width=\textwidth]{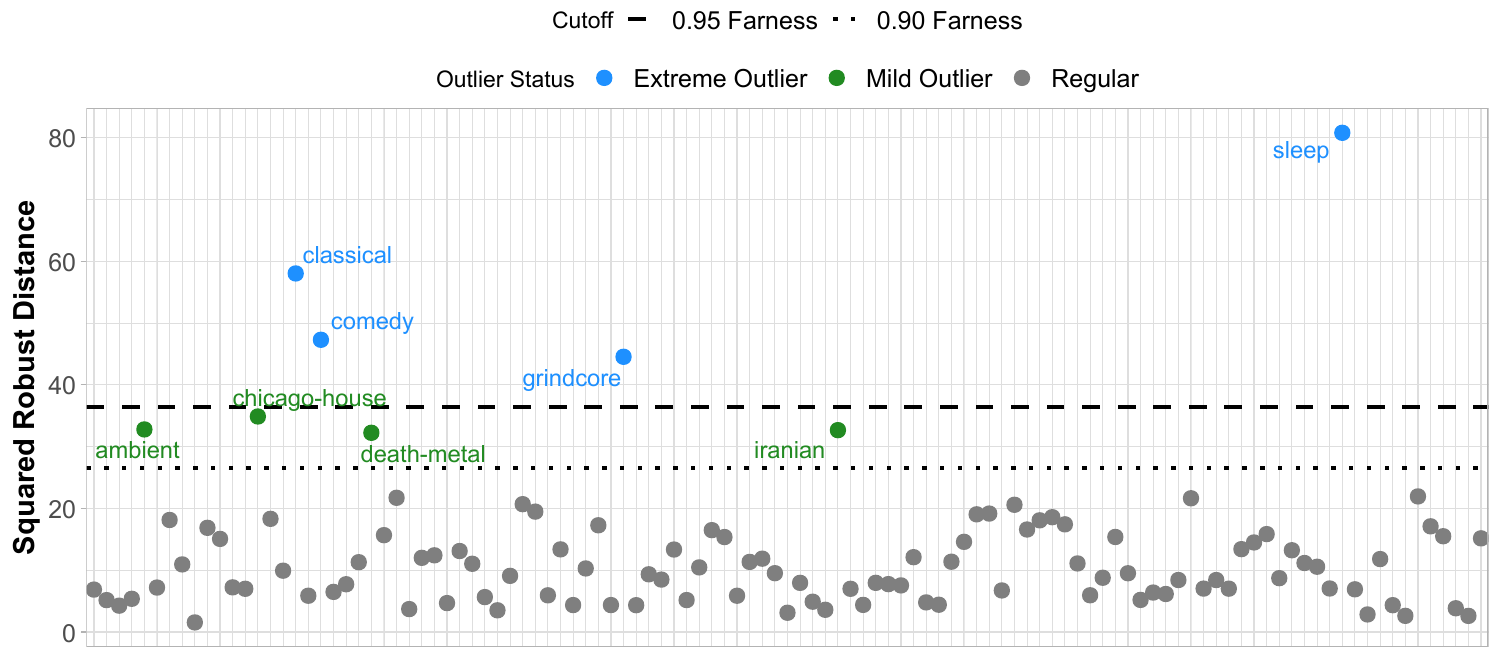}
    \caption{Plot of the squared robust Interval-Mahalanobis distances for the Spotify dataset. The horizontal lines correspond to the $0.9$ and $0.95$ farness cutoff values, and the extreme and mild outliers are marked in blue and green, respectively.}
    \label{fig:spotify_dist}
\end{figure}
 
\autoref{fig:spotify_shapley_barplot} shows the barplots of the Shapley value decomposition computed for the $12$ observations with the largest Interval-Mahalanobis distances. As indicated by the cutoff values represented by the dotted and dashed lines, the first four observations correspond to the extreme outliers and the next four to the mild outliers. We see that for \textit{sleep} and \textit{ambient} the features that most contribute to their outlyingness are \textit{loudness} and \textit{valence}, while \textit{energy} plays a stabilizing effect. This is consistent with the fact that these musical genres are usually less energetic and more calming. In the case of \textit{classical}, the main contributing features are \textit{duration}, \textit{loudness}, and \textit{acousticness}, which is expected since classical music tracks tend to be longer, less loud, and more acoustic compared to other genres. For \textit{comedy}, the main contributing features are \textit{tempo} and \textit{speechiness}, which is agreeable with comedy tracks having more spoken words and a different rhythm compared to regular music tracks. For \textit{grindcore}, the features with the highest contributions are \textit{danceability}, \textit{energy}, and \textit{acousticness}, which is consistent with the genre's characteristics of being fast-paced and intense. Finally, for \textit{chicago-house}, \textit{iranian}, and \textit{death-metal}, the contributions are more evenly distributed among the features, indicating that their outlyingness is due to a combination of factors rather than a single dominant feature.

\begin{figure}[ht]
    \centering
    \begin{subfigure}[t]{0.49\textwidth}
        \centering
        \includegraphics[width=\textwidth]{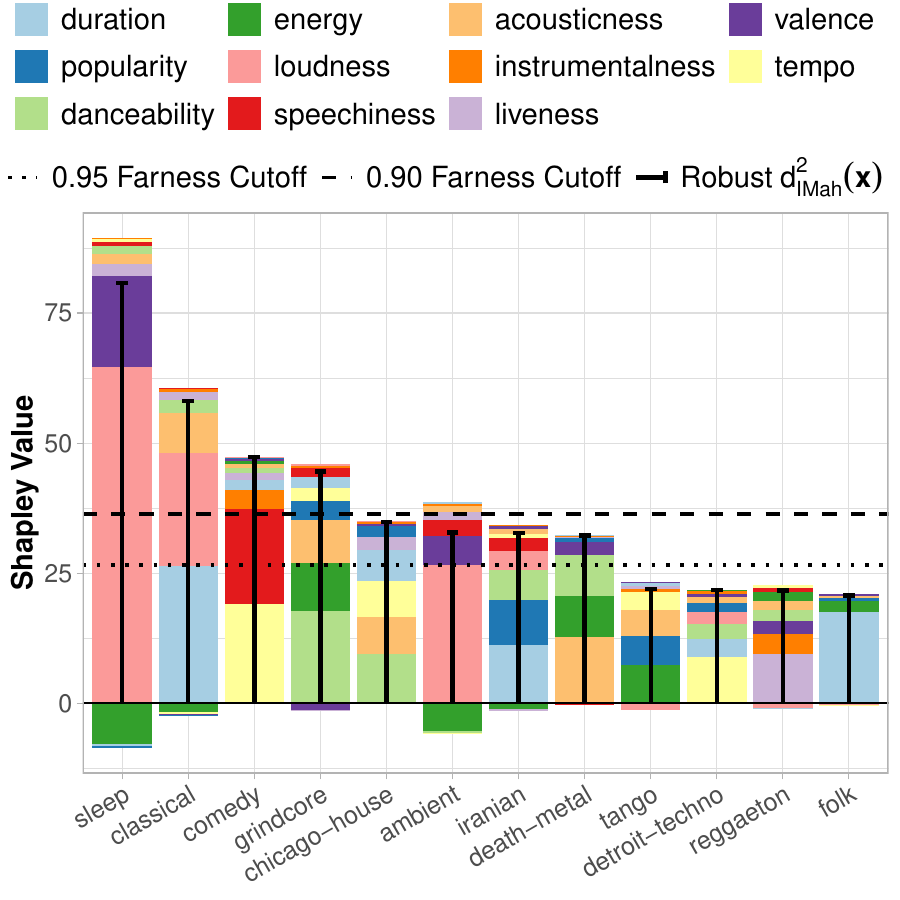}
        \caption{Barplot of the Shapley values.}
        \label{fig:spotify_shapley_barplot}
    \end{subfigure}
    \hfill
    \begin{subfigure}[t]{0.49\textwidth}
        \centering
        \includegraphics[width=\textwidth]{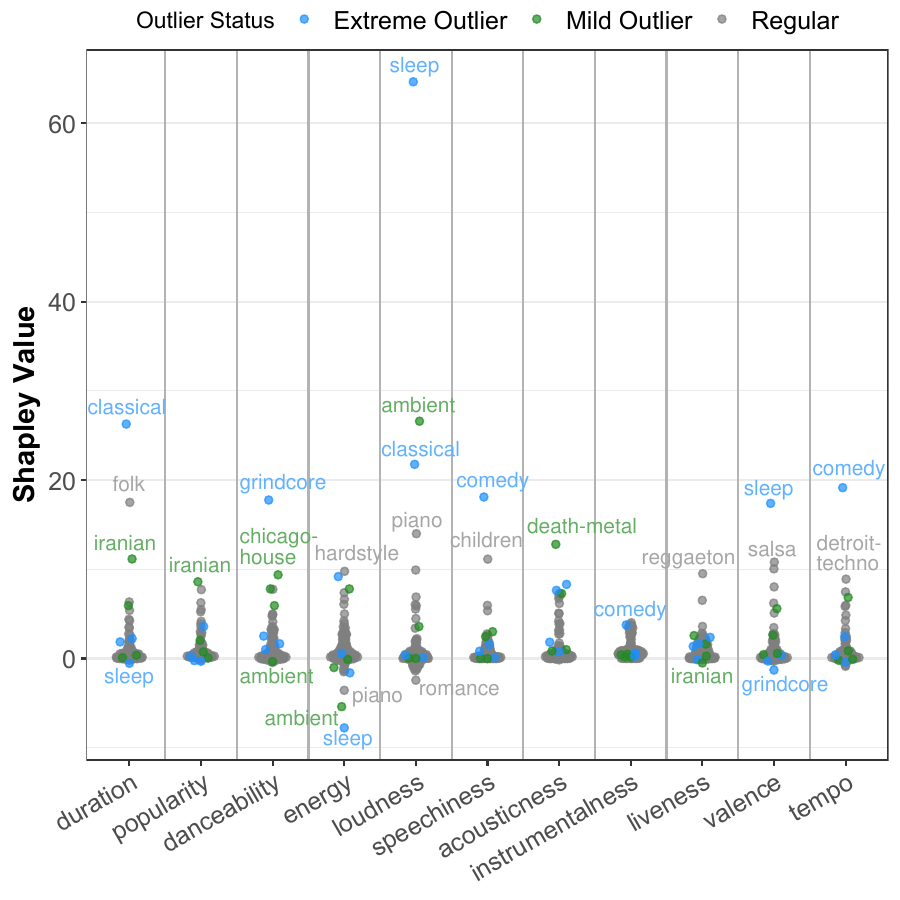}
        \caption{Beeswarm plot of the Shapley values.}
        \label{fig:spotify_shapley_beeswarm}
    \end{subfigure}
    \caption{Barplot (a) and beeswarm plot (b) of the Shapley value decomposition of the squared robust Interval-Mahalanobis distances for the Spotify dataset. In the barplot (a), only the $12$ observations with the largest distances are shown, ordered from greatest to smallest distance represented by the black vertical line, with the dotted and dashed lines representing the $0.9$ and $0.95$ farness cutoff values, respectively. In the beeswarm plot (b), the extreme outliers are marked in blue and the mild ones in green.}
    \label{fig:spotify_shapley}
\end{figure}

A beeswarm plot of the Shapley values for all observations per variable is plotted in \autoref{fig:spotify_shapley_beeswarm}, which allows us to assess how the outliers compare to regular observations and have a bigger picture regarding each variable's contribution. For instance, \textit{instrumentalness} stands out as the feature that, in general, contributes less to the Interval-Mahalanobis distance, while \textit{energy} is the feature that more often plays a stabilizing effect. On the contrary, \textit{loudness} seems to be one of the features that contributes the most to the final distance and whose observations with the highest Shapley values are outliers in the multivariate sense. It is also noteworthy that, for \textit{energy} and \textit{liveness}, the observations \textit{hardstyle} and \textit{reggaeton}, respectively, have the highest Shapley values. Therefore, these observations are candidates to cellwise outliers for these features, even though they are not multivariate outliers. Another example of a potential cellwise outlier is \textit{folk}, with the second-highest Shapley value for \textit{duration}, despite not being an outlier overall. This highlights the importance of looking at the Shapley values from different perspectives, as they can provide insights into the data that may not be apparent when looking at the overall distance alone.

In this case, the decomposition of the Shapley values into contributions of centers, ranges, and cross centers-ranges, actually consists of all three terms. This decomposition is shown in \autoref{fig:spotify_shapley_decomp} for the eight outlying observations. We see that, for all observations, the contribution of the centers is always the greatest, as expected, when taking into account the expression of the distance and the Shapley values. The contribution of the ranges is generally smaller, but, unlike in the Cars dataset, the difference is not as steep. For instance, the ranges' contribution for \textit{sleep}, \textit{classical}, and \textit{comedy} are considerable and increase the outlyingness of these genres. This suggests that the variability in the features also plays an important role in the outlyingness of these genres. In addition, the cross centers-ranges contributions are mostly negative for all eight observations, playing a stabilizing effect on the outlyingness of these genres. This indicates that there are interactions between the centers and ranges that mitigate the overall distance.

\begin{figure}[ht]
    \centering
    \includegraphics[width=\textwidth]{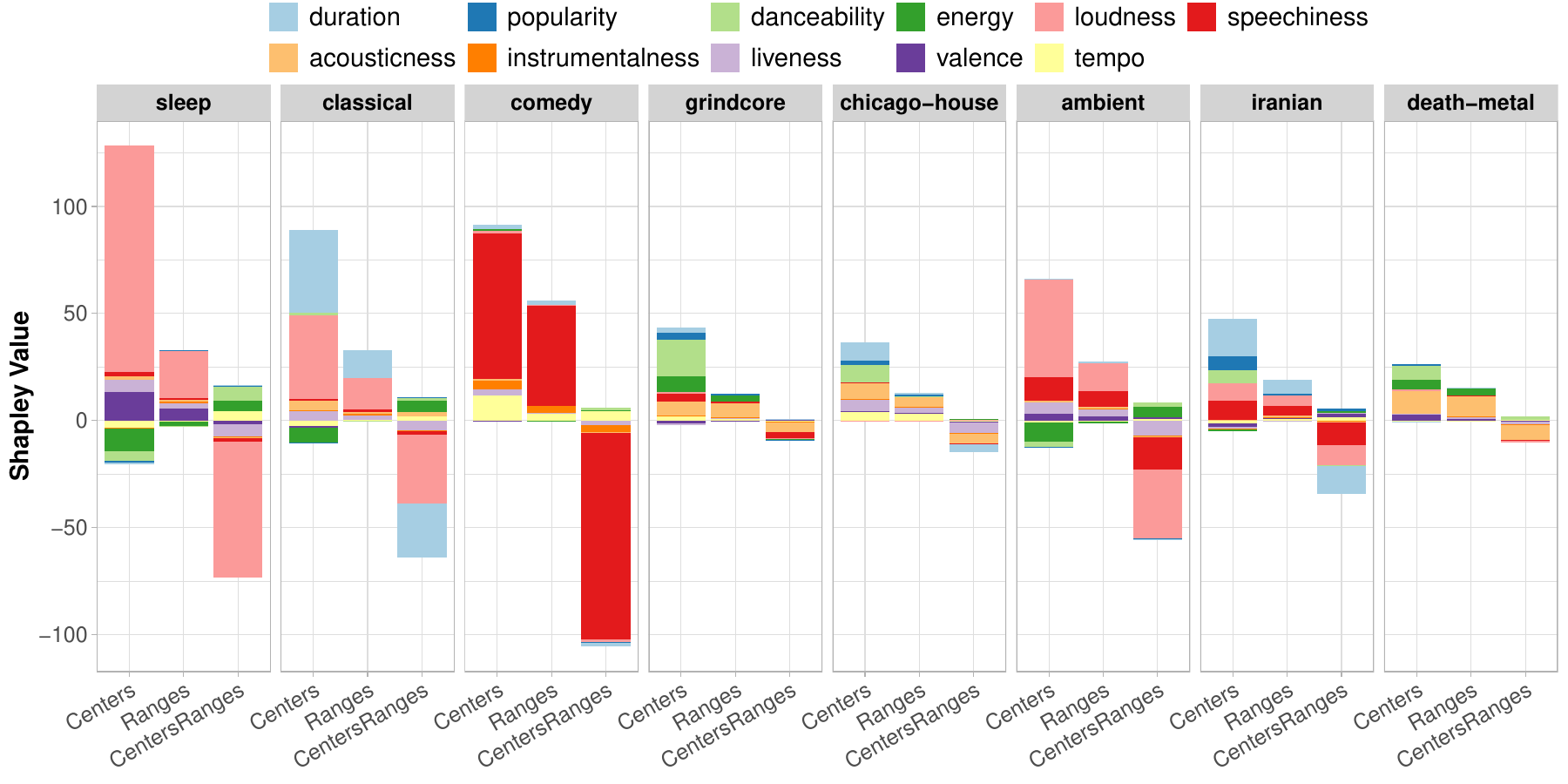}
    \caption{Barplot of the decomposition of the Shapley values by centers, ranges, and cross centers-ranges for the Spotify dataset. The four extreme outliers and the four mild outliers are shown, ordered from greatest to smallest Interval-Mahalanobis distance.}
    \label{fig:spotify_shapley_decomp}
\end{figure}

Overall, the Shapley value decomposition provides a detailed breakdown of the factors contributing to the outlyingness of each genre, allowing for a more nuanced understanding of the data and the characteristics that make certain genres stand out as outliers. In this particular case, the genres identified as outliers, such as \textit{sleep} and \textit{comedy}, are known to have distinct musical characteristics, with the Shapley values providing insights that align with musical knowledge about these genres.

\section{Conclusion}
\label{sec:conclusion}
In this work, we propose a novel framework for outlier explainability in interval-valued data, based on the Shapley value, a concept from cooperative game theory. We derive a closed-form expression for the Shapley values of the squared robust Interval-Mahalanobis distance, which is a measure of outlyingness for interval-valued data. The Shapley value is an additive decomposition of this distance into variable contributions, and of low computational cost. This is particularly useful for understanding which features are driving the outlyingness and, consequently, why an observation is flagged as an outlier.

The Shapley value decomposition can be further broken down into contributions of the centers, ranges, and cross centers-ranges, providing a more detailed understanding of the factors contributing to the outlyingness of an observation. Another interesting aspect of the Shapley value is its connection to the notion of cellwise outliers, helping to identify candidates to cellwise outliers and understand their input to the overall distance. We also derive the Shapley interaction index, which is a generalization of the Shapley value that allows attributing the outlyingness of an observation to pairs of variables. This can provide insights into interactions between variables that may not be apparent when looking at individual contributions alone.

All these potential uses of the Shapley value are illustrated through two real-world datasets from the automotive and music domains. In both cases, the method is able to provide clear, variable level explanations for the flagged observations, consistent with the relevant domain knowledge. The Shapley value decomposition allows us to understand the part each variable plays in the outlyingness of an observation, while the Shapley interaction index provides insights into interactions between variables.

No prior work, to our knowledge, has addressed explainability for interval-valued data, making this the first venture into this topic. We expect it can be a starting point for the development of more complex explainability methods, even outside the context of outlier detection, for interval-valued data or other symbolic data types.

\section*{Acknowledgments} 
\noindent Generative AI tools (GPT-5) were used for language editing, with the authors reviewing and editing the content as needed, and taking full responsibility for the content of the published article.
This work was supported by FCT — Funda\c{c}\~ao para a Ci\^encia e Tecnologia, I.P. through the grant \href{https://doi.org/10.54499/UI/BD/153720/2021}{UI/\discretionary{}{}{}BD/\discretionary{}{}{}153720/\discretionary{}{}{}2021}, and through the projects \href{https://doi.org/10.54499/UID/PRR/04621/2025}{UID/\discretionary{}{}{}04621/\discretionary{}{}{}2025}, \href{https://doi.org/10.54499/UID/50014/2025}{UID/\discretionary{}{}{}50014/\discretionary{}{}{}2025}, and UID/\discretionary{}{}{}4459/\discretionary{}{}{}2025.

\section*{Software Availability}\label{data-availability}
\noindent The proposed methodology and applications are implemented in the R package \texttt{AIDA}, available at \url{https://github.com/catarinaploureiro/AIDA}. The datasets used in this study are publicly available from CRAN and Kaggle as cited in the manuscript, and included in the R package \texttt{AIDA}. Specifically, the folder \texttt{raw-data} contains the pre-processing of the datasets, while the vignette \texttt{Shapley\_examples} replicates their analysis.

\small{
\bibliographystyle{elsarticle-harv}
\bibliography{cas-refs}}

\normalsize
\appendix
\section{Proofs from Section \ref{sec:explainability}}
\label{sec:appendix_explainability}

\subsection{Proof of \autoref{thm:shapley}}
\label{sec:proof_shapley}
Observing that $\mathcal{E}_{j\ell}=\mathcal{E}_{\ell j}=[\boldsymbol{\mathfrak{E}}_{UU}]_{j\ell}$ (by symmetry), $[\boldsymbol{\Sigma}_B^{-1}\boldsymbol{\Psi}]_{j\ell}=\omega_{j\ell}\psi_{\ell\ell}$, and $[\boldsymbol{\Psi}\boldsymbol{\Sigma}_B^{-1}]_{j\ell}=\psi_{jj}\omega_{j\ell}$, with $\psi_{j\ell}=\psi_{\ell j}=[\boldsymbol{\Psi}]_{j\ell}$ (by symmetry) and $\omega_{\ell j}=\omega_{j\ell}=[\boldsymbol{\Sigma}_B^{-1}]_{j\ell}$ (by symmetry), we can rewrite $d_\mathrm{IMah}^2(\boldsymbol{x})$ defined in \eqref{eq:d2_mah} as
\begin{equation}
\begin{aligned}
    d_\mathrm{IMah}^2(\boldsymbol{x})=&\sum_{j=1}^p\sum_{\ell=1}^p\biggl[(c_j-\mu_{C_j})(c_\ell-\mu_{C_\ell})\omega_{j\ell}+\frac{1}{4}(r_j-\mu_{C_j})(r_\ell-\mu_{R_\ell})\mathcal{E}_{j\ell}\omega_{j\ell} \\
    &\qquad\qquad+ (c_j-\mu_{C_j})(r_\ell-\mu_{R_\ell})\omega_{j\ell}\psi_{\ell\ell}\biggr].
\end{aligned}
\end{equation}
 
Then, for any subset of variables $S\subseteq P\setminus\{j\}$, we can write 
\begin{align}
    \label{eq:marginal_contribution}
    \Delta_jd_\mathrm{IMah}^2(\boldsymbol{x}^S)&=d_\mathrm{IMah}^2(\boldsymbol{x}^{S\cup\{j\}})-d_\mathrm{IMah}^2(\boldsymbol{x}^S)\\
    &=\sum_{\ell\in S\cup\{j\}}\sum_{h\in S\cup\{j\}}\biggl[(c_\ell-\mu_{C_\ell})(c_h-\mu_{C_h})\omega_{\ell h}+\frac{1}{4}(r_\ell-\mu_{R_\ell})(r_h-\mu_{R_h})\mathcal{E}_{\ell h}\omega_{\ell h} \nonumber\\
    &\qquad\qquad\qquad\qquad+ (c_\ell-\mu_{C_\ell})(r_h-\mu_{R_h})\omega_{\ell h}\psi_{hh}\biggr]\nonumber\\
    &\quad-\sum_{\ell\in S}\sum_{h\in S}\biggl[(c_\ell-\mu_{C_\ell})(c_h-\mu_{C_h})\omega_{\ell h}+\frac{1}{4}(r_\ell-\mu_{R_\ell})(r_h-\mu_{R_h})\mathcal{E}_{\ell h}\omega_{\ell h} \nonumber\\
    &\qquad\qquad\qquad+ (c_\ell-\mu_{C_\ell})(r_h-\mu_{R_h})\omega_{\ell h}\psi_{hh}\biggr]\nonumber\\
    &=\sum_{\ell\in S}\biggl[(c_\ell-\mu_{C_\ell})(c_j-\mu_{C_j})\omega_{\ell j}+\frac{1}{4}(r_\ell-\mu_{R_\ell})(r_j-\mu_{R_j})\mathcal{E}_{\ell j}\omega_{\ell j} \nonumber\\
    &\qquad\qquad+ (c_\ell-\mu_{C_\ell})(r_j-\mu_{R_j})\omega_{\ell j}\psi_{jj}\biggr]\nonumber\\
    &\quad+\sum_{h\in S}\biggl[(c_j-\mu_{C_j})(c_h-\mu_{C_h})\omega_{jh}+\frac{1}{4}(r_j-\mu_{R_j})(r_h-\mu_{R_h})\mathcal{E}_{jh}\omega_{jh} \nonumber\\
    &\qquad\qquad+ (c_j-\mu_{C_j})(r_h-\mu_{R_h})\omega_{jh}\psi_{hh}\biggr]\nonumber\\
    &\quad+(c_j-\mu_{C_j})^2\omega_{jj}+\frac{1}{4}(r_j-\mu_{R_j})^2\mathcal{E}_{jj}\omega_{jj}+(c_j-\mu_{C_j})(r_j-\mu_{R_j})\omega_{jj}\psi_{jj}\nonumber\\
    \label{eq:marginal_contribution2}
    &=2(c_j-\mu_{C_j})\sum_{\ell\in S}(c_\ell-\mu_{C_\ell})\omega_{j\ell}+\frac{1}{2}(r_j-\mu_{R_j})\sum_{\ell\in S}(r_\ell-\mu_{R_\ell})\mathcal{E}_{j\ell}\omega_{j\ell}\nonumber\\
    &\quad+(r_j-\mu_{R_j})\psi_{jj}\sum_{\ell\in S}(c_\ell-\mu_{C_\ell})\omega_{j\ell}+(c_j-\mu_{C_j})\sum_{\ell\in S}(r_\ell-\mu_{R_\ell})\psi_{\ell\ell}\omega_{j\ell}\\
    &\quad+(c_j-\mu_{C_j})^2\omega_{jj}+\frac{1}{4}(r_j-\mu_{R_j})^2\mathcal{E}_{jj}\omega_{jj}+(c_j-\mu_{C_j})(r_j-\mu_{R_j})\omega_{jj}\psi_{jj}.\nonumber
\end{align}

Defining $w(|S|)=\dfrac{|S|!(p-|S|-1)!}{p!}$, where $|\cdot|$ denotes set cardinality, we can derive the Shapley value of the $j$-th variable as
\begin{align*}
    \phi_j(\boldsymbol{x})&=\sum_{S\subseteq P\setminus\{j\}}w(|S|)\Delta_jd_\mathrm{IMah}^2(\boldsymbol{x}^S)\\
    &=\sum_{S\subseteq P\setminus\{j\}}w(|S|)\left[2(c_j-\mu_{C_j})\sum_{\ell\in S}(c_\ell-\mu_{C_\ell})\omega_{j\ell}+\frac{1}{2}(r_j-\mu_{R_j})\sum_{\ell\in S}(r_\ell-\mu_{R_\ell})\mathcal{E}_{j\ell}\omega_{j\ell}\right.\\
    &\qquad\qquad\qquad\qquad+(r_j-\mu_{R_j})\psi_{jj}\sum_{\ell\in S}(c_\ell-\mu_{C_\ell})\omega_{j\ell}+(c_j-\mu_{C_j})\sum_{\ell\in S}(r_\ell-\mu_{R_\ell})\psi_{\ell\ell}\omega_{j\ell}\\
    &\qquad\qquad\qquad\qquad+\left.(c_j-\mu_{C_j})^2\omega_{jj}+\frac{1}{4}(r_j-\mu_{R_j})^2\mathcal{E}_{jj}\omega_{jj}+(c_j-\mu_{C_j})(r_j-\mu_{R_j})\omega_{jj}\psi_{jj}\right]\\
    &=\sum_{S\subseteq P\setminus\{j\}}w(|S|)\left[2(c_j-\mu_{C_j})\sum_{\ell\in S}(c_\ell-\mu_{C_\ell})\omega_{j\ell}+\frac{1}{2}(r_j-\mu_{R_j})\sum_{\ell\in S}(r_\ell-\mu_{R_\ell})\mathcal{E}_{j\ell}\omega_{j\ell}\right.\\
    &\qquad\qquad\qquad\qquad+\left.(r_j-\mu_{R_j})\psi_{jj}\sum_{\ell\in S}(c_\ell-\mu_{C_\ell})\omega_{j\ell}+(c_j-\mu_{C_j})\sum_{\ell\in S}(r_\ell-\mu_{R_\ell})\psi_{\ell\ell}\omega_{j\ell}\right]\\
    &\quad+\sum_{S\subseteq P\setminus\{j\}}w(|S|)\left[(c_j-\mu_{C_j})^2\omega_{jj}+\frac{1}{4}(r_j-\mu_{R_j})^2\mathcal{E}_{jj}\omega_{jj}+(c_j-\mu_{C_j})(r_j-\mu_{R_j})\omega_{jj}\psi_{jj}\right].
\end{align*}
Since there are $\dbinom{p-1}{s}$ subsets of $P\setminus\{j\}$ of size $s$, we have
\begin{equation}
    \sum\limits_{S\subseteq P\setminus\{j\}}w(|S|)=\sum\limits_{s=0}^{p-1}\dbinom{p-1}{s}\dfrac{s!(p-s-1)!}{p!}=1,
\end{equation}
where $s=|S|$ and $\dbinom{n}{k}=\dfrac{n!}{k!(n-k)!}$ is the binomial coefficient, we can further simplify the expression as
\begin{align*}
    \phi_j(\boldsymbol{x})&=\sum_{s=1}^{p-1}\Biggl(w(s)\sum_{\substack{S\subseteq P\setminus\{j\}\\|S|=s}}\left[2(c_j-\mu_{C_j})\sum_{\ell\in S}(c_\ell-\mu_{C_\ell})\omega_{j\ell}+\frac{1}{2}(r_j-\mu_{R_j})\sum_{\ell\in S}(r_\ell-\mu_{R_\ell})\mathcal{E}_{j\ell}\omega_{j\ell}\right.\\
    &\qquad\qquad\qquad\qquad+\left.(r_j-\mu_{R_j})\psi_{jj}\sum_{\ell\in S}(c_\ell-\mu_{C_\ell})\omega_{j\ell}+(c_j-\mu_{C_j})\sum_{\ell\in S}(r_\ell-\mu_{R_\ell})\psi_{\ell\ell}\omega_{j\ell}\right]\Biggr)\\
    &\quad+(c_j-\mu_{C_j})^2\omega_{jj}+\frac{1}{4}(r_j-\mu_{R_j})^2\mathcal{E}_{jj}\omega_{jj}+(c_j-\mu_{C_j})(r_j-\mu_{R_j})\omega_{jj}\psi_{jj}.
\end{align*}
Since for each fixed $\ell$, the number of subsets $S\subseteq P\setminus\{j\}$ with $|S|=s$ and $\ell\in S$ is $\dbinom{p-2}{s-1}$, we have
\begin{align*}
    \phi_j(\boldsymbol{x})&=\sum_{s=1}^{p-1}\Biggl(w(s)\binom{p-2}{s-1}\left[2(c_j-\mu_{C_j})\sum_{\ell\in S}(c_\ell-\mu_{C_\ell})\omega_{j\ell}+\frac{1}{2}(r_j-\mu_{R_j})\sum_{\ell\in S}(r_\ell-\mu_{R_\ell})\mathcal{E}_{j\ell}\omega_{j\ell}\right.\\
    &\qquad\qquad\qquad\qquad+\left.(r_j-\mu_{R_j})\psi_{jj}\sum_{\ell\in S}(c_\ell-\mu_{C_\ell})\omega_{j\ell}+(c_j-\mu_{C_j})\sum_{\ell\in S}(r_\ell-\mu_{R_\ell})\psi_{\ell\ell}\omega_{j\ell}\right]\Biggr)\\
    &\quad+(c_j-\mu_{C_j})^2\omega_{jj}+\frac{1}{4}(r_j-\mu_{R_j})^2\mathcal{E}_{jj}\omega_{jj}+(c_j-\mu_{C_j})(r_j-\mu_{R_j})\omega_{jj}\psi_{jj},
\end{align*}
and noticing that $w(s)\dbinom{p-2}{s-1}=\dfrac{s}{p(p-1)}$, it becomes
\begin{align*}
    \phi_j(\boldsymbol{x})&=\sum_{s=1}^{p-1}\Biggl(\frac{s}{p(p-1)}\left[2(c_j-\mu_{C_j})\sum_{\ell\in S}(c_\ell-\mu_{C_\ell})\omega_{j\ell}+\frac{1}{2}(r_j-\mu_{R_j})\sum_{\ell\in S}(r_\ell-\mu_{R_\ell})\mathcal{E}_{j\ell}\omega_{j\ell}\right.\\
    &\qquad\qquad\qquad\qquad+\left.(r_j-\mu_{R_j})\psi_{jj}\sum_{\ell\in S}(c_\ell-\mu_{C_\ell})\omega_{j\ell}+(c_j-\mu_{C_j})\sum_{\ell\in S}(r_\ell-\mu_{R_\ell})\psi_{\ell\ell}\omega_{j\ell}\right]\Biggr)\\
    &\quad+(c_j-\mu_{C_j})^2\omega_{jj}+\frac{1}{4}(r_j-\mu_{R_j})^2\mathcal{E}_{jj}\omega_{jj}+(c_j-\mu_{C_j})(r_j-\mu_{R_j})\omega_{jj}\psi_{jj}.
\end{align*}
Finally, since $\sum\limits_{s=1}^{p-1}\dfrac{s}{p(p-1)}=\dfrac{1}{2}$, we obtain
\begin{align*}
    \phi_j(\boldsymbol{x})&=(c_j-\mu_{C_j})\sum_{\ell\in P\setminus\{j\}}(c_\ell-\mu_{C_\ell})\omega_{j\ell}+\frac{1}{4}(r_j-\mu_{R_j})\sum_{\ell\in P\setminus\{j\}}(r_\ell-\mu_{R_\ell})\mathcal{E}_{j\ell}\omega_{j\ell}\\
    &\quad+\frac{1}{2}(r_j-\mu_{R_j})\psi_{jj}\sum_{\ell\in P\setminus\{j\}}(c_\ell-\mu_{C_\ell})\omega_{j\ell}+\frac{1}{2}(c_j-\mu_{C_j})\sum_{\ell\in P\setminus\{j\}}(r_\ell-\mu_{R_\ell})\psi_{\ell\ell}\omega_{j\ell}\\
    &\quad+(c_j-\mu_{C_j})^2\omega_{jj}+\frac{1}{4}(r_j-\mu_{R_j})^2\mathcal{E}_{jj}\omega_{jj}+(c_j-\mu_{C_j})(r_j-\mu_{R_j})\omega_{jj}\psi_{jj}\\
    &=(c_j-\mu_{C_j})\sum_{\ell=1}^p(c_\ell-\mu_{C_\ell})\omega_{j\ell}+\frac{1}{4}(r_j-\mu_{R_j})\sum_{\ell=1}^p(r_\ell-\mu_{R_\ell})\mathcal{E}_{j\ell}\omega_{j\ell}\\
    &\quad+\frac{1}{2}(r_j-\mu_{R_j})\psi_{jj}\sum_{\ell=1}^p(c_\ell-\mu_{C_\ell})\omega_{j\ell}+\frac{1}{2}(c_j-\mu_{C_j})\sum_{\ell=1}^p(r_\ell-\mu_{R_\ell})\psi_{\ell\ell}\omega_{j\ell}.
\end{align*}
This concludes the proof.
\qed

\subsection{Proof of \autoref{thm:shapley_interaction}}
\label{sec:proof_shapley_interaction}
Considering the marginal contribution of the $j$-th variable to the squared distance obtained in the proof of \autoref{thm:shapley} given by \eqref{eq:marginal_contribution} and \eqref{eq:marginal_contribution2}, we can obtain the marginal contribution of the interaction between the $j$-th and $\ell$-th variables as
\begin{align*}
    \Delta_{\{j,\ell\}}d_\mathrm{IMah}^2(\boldsymbol{x}^T)&=\Delta_\ell d_\mathrm{IMah}^2(\boldsymbol{x}^{T\cup\{j\}})-\Delta_\ell d_\mathrm{IMah}^2(\boldsymbol{x}^T)\\
    &=\left[d_\mathrm{IMah}^2(\boldsymbol{x}^{T\cup\{j,\ell\}})-d_\mathrm{IMah}^2(\boldsymbol{x}^{T\cup\{j\}})\right]-\left[d_\mathrm{IMah}^2(\boldsymbol{x}^{T\cup\{\ell\}})-d_\mathrm{IMah}^2(\boldsymbol{x}^T)\right]\\
    &=\Bigg[2(c_\ell-\mu_{C_\ell})\sum_{h\in T\cup\{j\}}(c_h-\mu_{C_h})\omega_{h\ell}+\frac{1}{2}(r_\ell-\mu_{R_\ell})\sum_{h\in T\cup\{j\}}(r_h-\mu_{R_h})\mathcal{E}_{h\ell}\omega_{h\ell}\Bigg.\\
    &\qquad+(r_\ell-\mu_{R_\ell})\psi_{\ell\ell}\sum_{h\in T\cup\{j\}}(c_h-\mu_{C_h})\omega_{h\ell}+(c_\ell-\mu_{C_\ell})\sum_{h\in T\cup\{j\}}(r_h-\mu_{R_h})\psi_{hh}\omega_{h\ell}\\
    &\qquad+\Bigg.(c_\ell-\mu_{C_\ell})^2\omega_{\ell\ell}+\frac{1}{4}(r_\ell-\mu_{R_\ell})^2\mathcal{E}_{\ell\ell}\omega_{\ell\ell}+(c_\ell-\mu_{C_\ell})(r_\ell-\mu_{R_\ell})\omega_{\ell\ell}\psi_{\ell\ell}\Bigg]\\
    &\ \quad-\Bigg[2(c_\ell-\mu_{C_\ell})\sum_{h\in T}(c_h-\mu_{C_h})\omega_{h\ell}+\frac{1}{2}(r_\ell-\mu_{R_\ell})\sum_{h\in T}(r_h-\mu_{R_h})\mathcal{E}_{h\ell}\omega_{h\ell}\Bigg.\\
    &\quad\qquad+(r_\ell-\mu_{R_\ell})\psi_{\ell\ell}\sum_{h\in T}(c_h-\mu_{C_h})\omega_{h\ell}+(c_\ell-\mu_{C_\ell})\sum_{h\in T}(r_h-\mu_{R_h})\psi_{hh}\omega_{h\ell}\\
    &\quad\qquad+\Bigg.(c_\ell-\mu_{C_\ell})^2\omega_{\ell\ell}+\frac{1}{4}(r_\ell-\mu_{R_\ell})^2\mathcal{E}_{\ell\ell}\omega_{\ell\ell}+(c_\ell-\mu_{C_\ell})(r_\ell-\mu_{R_\ell})\omega_{\ell\ell}\psi_{\ell\ell}\Bigg]\\
    &=2(c_j-\mu_{C_j})(c_\ell-\mu_{C_\ell})\omega_{j\ell}+\frac{1}{2}(r_j-\mu_{R_j})(r_\ell-\mu_{R_\ell})\mathcal{E}_{j\ell}\omega_{j\ell}\\
    &\quad+(c_j-\mu_{C_j})(r_\ell-\mu_{R_\ell})\psi_{\ell\ell}\omega_{j\ell}+(r_j-\mu_{R_j})(c_\ell-\mu_{C_\ell})\psi_{jj}\omega_{j\ell}.\\
\end{align*}

The Shapley value of the interaction between the $j$-th and $\ell$-th variables, for $j\neq\ell$, is then given by
\begin{align*}
    \Phi_{j\ell}(\boldsymbol{x})&=\sum_{T\subseteq P\setminus\{j,\ell\}}\frac{|T|!(p-|T|-2)!}{(p-1)!}\Delta_{\{j,\ell\}}d_\mathrm{IMah}^2(\boldsymbol{x}^T)\\
    &=\sum_{T\subseteq P\setminus\{j,\ell\}}\frac{|T|!(p-|T|-2)!}{(p-1)!}\Bigg[2(c_j-\mu_{C_j})(c_\ell-\mu_{C_\ell})\omega_{j\ell}+\frac{1}{2}(r_j-\mu_{R_j})(r_\ell-\mu_{R_\ell})\mathcal{E}_{j\ell}\omega_{j\ell}\Bigg.\\
    &\qquad\qquad\qquad\qquad\qquad\qquad+\Bigg.(c_j-\mu_{C_j})(r_\ell-\mu_{R_\ell})\psi_{\ell\ell}\omega_{j\ell}+(r_j-\mu_{R_j})(c_\ell-\mu_{C_\ell})\psi_{jj}\omega_{j\ell}\Bigg],
\end{align*}
and since there are $\dbinom{p-2}{t}$ subsets of size $t$ in $P\setminus{j,\ell}$, we have
\begin{equation*}
    \sum_{T\subseteq P\setminus\{j,\ell\}}\dfrac{|T|!(p-|T|-2)!}{(p-1)!}=\sum\limits_{t=0}^{p-2}\binom{p-2}{t}\frac{t!(p-t-2)!}{(p-1)!}=1,
\end{equation*}
where $t=|T|$, which leads to
\begin{align*}
    \Phi_{j\ell}(\boldsymbol{x})&=2(c_j-\mu_{C_j})(c_\ell-\mu_{C_\ell})\omega_{j\ell}+\frac{1}{2}(r_j-\mu_{R_j})(r_\ell-\mu_{R_\ell})\mathcal{E}_{j\ell}\omega_{j\ell}\\
    &\quad+\Bigg.(c_j-\mu_{C_j})(r_\ell-\mu_{R_\ell})\psi_{\ell\ell}\omega_{j\ell}+(r_j-\mu_{R_j})(c_\ell-\mu_{C_\ell})\psi_{jj}\omega_{j\ell}.
\end{align*}

As for the diagonal elements of the interaction matrix, combining the similar terms in $\phi_j(\boldsymbol{x})$ and $\sum_{\ell\neq j}\Phi_{j\ell}(\boldsymbol{x})$, we can write that
\begin{align*}
    \Phi_{jj}(\boldsymbol{x})&=\phi_j(\boldsymbol{x})-\sum_{\ell\neq j}\Phi_{j\ell}(\boldsymbol{x})\\
    &=(c_j-\mu_{C_j})\sum_{\ell=1}^p(c_\ell-\mu_{C_\ell})\omega_{j\ell}-2(c_j-\mu_{C_j})\sum_{\ell\neq j}(c_\ell-\mu_{C_\ell})\omega_{j\ell}\\
    &\quad+\frac{1}{4}(r_j-\mu_{R_j})\sum_{\ell=1}^p(r_\ell-\mu_{R_\ell})\mathcal{E}_{j\ell}\omega_{j\ell}-\frac{1}{2}(r_j-\mu_{R_j})\sum_{\ell\neq j}(r_\ell-\mu_{R_\ell})\mathcal{E}_{j\ell}\omega_{j\ell}\\
    &\quad+\frac{1}{2}(c_j-\mu_{C_j})\sum_{\ell=1}^p(r_\ell-\mu_{R_\ell})\psi_{\ell\ell}\omega_{j\ell}-(c_j-\mu_{C_j})\sum_{\ell\neq j}(r_\ell-\mu_{R_\ell})\psi_{\ell\ell}\omega_{j\ell}\\
    &\quad+\frac{1}{2}(r_j-\mu_{R_j})\psi_{jj}\sum_{\ell=1}^p(c_\ell-\mu_{C_\ell})\omega_{j\ell}-(r_j-\mu_{R_j})\sum_{\ell\neq j}(c_\ell-\mu_{C_\ell})\psi_{jj}\omega_{j\ell}\\
    &=(c_j-\mu_{C_j})^2\omega_{jj}-(c_j-\mu_{C_j})\sum_{\ell\neq j}(c_\ell-\mu_{C_\ell})\omega_{j\ell}\\
    &\quad+\frac{1}{4}(r_j-\mu_{R_j})^2\mathcal{E}_{jj}\omega_{jj}-\frac{1}{4}(r_j-\mu_{R_j})\sum_{\ell\neq j}(r_\ell-\mu_{R_\ell})\mathcal{E}_{j\ell}\omega_{j\ell}\\
    &\quad+\frac{1}{2}(c_j-\mu_{C_j})(r_j-\mu_{R_j})\psi_{jj}\omega_{jj}-\frac{1}{2}(c_j-\mu_{C_j})\sum_{\ell\neq j}(r_\ell-\mu_{R_\ell})\psi_{\ell\ell}\omega_{j\ell}\\
    &\quad+\frac{1}{2}(r_j-\mu_{R_j})(c_j-\mu_{C_j})\psi_{jj}\omega_{jj}-\frac{1}{2}(r_j-\mu_{R_j})\psi_{jj}\sum_{\ell\neq j}(c_\ell-\mu_{C_\ell})\omega_{j\ell}.
\end{align*}

This concludes the proof of \autoref{thm:shapley_interaction}.
\qed

\end{document}